\documentclass[11pt, a4paper]{amsart}

\usepackage{amsmath, amsthm, amsfonts, amssymb}
\usepackage{mathtools}
\usepackage[all]{xy}
\usepackage{graphicx}
\usepackage{nicefrac}

\newcommand{\R}{\mathbb{R}}

\newcommand{\doslineas}[2]{\genfrac{}{}{0pt}{}{#1}{#2}}

\newtheorem{thm}{Theorem}[section]

\newtheorem{lem}[thm]{Lemma}

\theoremstyle{definition}
\newtheorem{defn}[thm]{Definition}

\theoremstyle{remark}
\newtheorem{remark}[thm]{Remark}

\numberwithin{equation}{section}

\begin{document}

\title[Solving Kepler's equation via Smale's $\alpha$-theory]
      {Solving Kepler's equation via Smale's $\alpha$-theory}

\author[M.~Avendano]{Mart\'\i n Avendano}
\address{Centro Universitario de la Defensa\\
         Academia General Militar\\
         Ctra. de Huesca s/n\\
         50090, Zaragoza, Spain and IUMA, Universidad de Zaragoza, Spain}
\email{avendano@unizar.es, vmartin@unizar.es, jortigas@unizar.es}
\author[V.~Mart\'\i n-Molina]{Ver\'onica Mart\'\i n-Molina}
\author[J.~Ortigas-Galindo]{Jorge Ortigas-Galindo}

\thanks{Second author is partially supported by the MINECO grant MTM2011-22621 and the FQM-327 group (Junta de Andaluc\'ia, Spain). Third author is partially supported by the MINECO grant MTM2010-21740-C02-02. Both are also partially supported by the Grupo consolidado E15 ``Geometr\'ia'' (Gobierno de Arag\'on, Spain) and the ``Centro Universitario de la Defensa de Zaragoza'' grant ID2013-15.}

\begin{abstract}
We obtain an approximate solution $\tilde{E}=\tilde{E}(e,M)$ of Kepler's equation
$E-e\sin(E)=M$ for any $e\in[0,1)$ and $M\in[0,\pi]$. Our solution is guaranteed,
via Smale's $\alpha$-theory, to converge to the actual solution $E$ through Newton's
method at quadratic speed, i.e. the $n$-th iteration produces a value $E_n$ such
that $|E_n-E|\leq (\frac12)^{2^n-1}|\tilde{E}-E|$. The formula provided for $\tilde{E}$
is a piecewise rational function with conditions defined by polynomial inequalities,
except for a small region near $e=1$ and $M=0$, where a single cubic root is used. We
also show that the root operation is unavoidable, by proving that no approximate solution
can be computed in the entire region $[0,1)\times[0,\pi]$ if only rational functions are
allowed in each branch.
\end{abstract}

\maketitle

\section{Introduction}\label{sec-intro}

Kepler's laws describe the way planets move in their orbits about the Sun.
Geometrically, they say that the planets move in planar elliptical orbits with
eccentricity $e\in[0,1)$, and that the area swept by the line joining the planet
and the Sun increases linearly with time, which leads immediately to Kepler's
equation $E-e\sin(E)=M$, relating mean and eccentric anomalies: the mean anomaly
is a fictitious angle $M$ that increases linearly with time at a rate $M=2\pi t/T$,
where $T$ is the orbital period, and the eccentric anomaly $E$ gives the coordinates
of the planet in its orbit plane as $(x,y)=(a\cos(E), b\sin(E))$. Here, the
$xy$-plane has origin at the center of the ellipse with the $x$-axis pointing to
the perihelion, and the values $a$ and $b$ are the semi-major and semi-minor axis
of the ellipse. Therefore, finding the exact location of a planet at a given time
requires solving an instance of Kepler's equation for some $M$, assuming that the
values $a$, $b$, $e$ and $T$ are known (actually, only $a$ and $e$ are needed,
since $b=a\sqrt{1-e^2}$ and $T$ can be obtained from $a$ using the third law).
For a derivation of these formulas, and a detailed introduction to Kepler's
equation, see~\cite{Batt}.

By a symmetry argument, the equation can be easily reduced to the case $M\in[0,\pi]$.
The existence and uniqueness of solution $E\in[0,\pi]$ follows from the fact that
the function $f_{e,M}:[0,\pi]\to[0,\pi]$ given by $f_{e,M}(E)=E-e\sin(E)-M$ is
strictly increasing.

Several solutions to the problem have been proposed since it was stated $400$ years
ago. Some authors have tried non-iterative methods to solve the equation up to a
fixed predetermined accuracy (\cite{Mark}; \cite{MC07}). However, we want to calculate
the solution with arbitrary precision, hence our interest in iterative techniques.

Kepler himself proposed to use a fixed-point iteration to solve the equation~(\cite[Ch.~1]{Colw}),
i.e.~guess $E_0$, an approximation of the exact solution $E$, and then iterate
$E_{n+1}=M+e\sin(E_n)$. This sequence converges to $E$, since
$|E_{n+1}-E|=|M+e\sin(E_n)-E|=e|\sin(E_n)-\sin(E)|\leq e|E_n-E|$, which implies that
$|E_n-E|\leq e^n|E_0-E|\longrightarrow 0$ as $n\to\infty$. The problem with this approach
is that the convergence is slow for values of $e$ near $1$. For the orbit of Mercury, which has
$e\approx 0.2$, about $5$ iterations are needed to reduce the error by a factor of
$10^{-3}$, while for values of eccentricity $e>0.5$ the fixed-point iteration is even
slower than a bisection method.

Although the fixed-point iteration does not provide an efficient solution to Kepler's
equation, it exhibits the structure of most of the current methods to solve it: first,
guess an approximation $\tilde{E}$ of the solution (called \emph{starter}), and then use some
iterative technique to produce a sequence quickly converging to the actual solution
(see \cite{Danb}, \cite{DB83}, \cite{ME13}, \cite{Pala}). For the second part, Newton's
method seems to be the most used iteration, mainly due to its conceptual simplicity,
generality and fast convergence. The guessing part, however, requires
some specific understanding on the equation and has been the subject of many recent
papers (\cite{CEMR}; \cite{Mikk}; \cite{Ng}; \cite{Nije}; \cite{OG86}; \cite{TB89}).

Starters have been compared (and optimized) using different criteria, such as
the number of iterations needed to reach certain precision, the distance to the
actual solution, the number of floating point operations needed for its
computation, etc. For this purpose, we adopt a criterion which is very specific
to Newton's method and guarantees that the iterations reduce the error at
quadratic speed. More precisely, we will only accept an approximate solution
$\tilde{E}$ of the equation $f_{e,M}(E)=0$ if Newton's method starting at
$E_0=\tilde{E}$ produces a sequence $E_n$ such that
$|E_n-E|\leq (\frac12)^{2^n-1}|\tilde{E}-E|$ for all $n\geq 0$.

Taking one of these starters satisfying $\tilde{E} \in [0,\pi ]$, the initial error is
at most $\pi$, so we obtain an accuracy $10^{-N}$ after only
$n= \lceil \log_2 \left( 1+\log_{2} (\pi)+ \log_{2} (10) N \right) \rceil$ iterations.
In particular, ten iterations of Newton's method starting from $\tilde{E}$ give an error
less than $10^{-307}$ for any input value of $e$ and $M$.

We will use a simple test, due to Smale~\cite{Smale} and later improved by Wang and
Han~\cite{WH89}, which depends only on the starter $\tilde{E}$, and guarantees the speed
of convergence that we claim.

\begin{defn}[Smale's $\alpha$-test]
We say that $\widetilde{E}$ is an \textbf{approximate zero} of $f_{e,M}$ if it
satisfies the following condition
\[
\alpha(f_{e,M},\widetilde E)=\beta(f_{e,M},\widetilde E) \cdot \gamma(f_{e,M},\widetilde  E) <\alpha_0,
\]
where
\[
\beta(f_{e,M},\widetilde{E})=\left|\frac{f_{e,M}(\widetilde{E})}{f_{e,M}'(\widetilde{E})} \right|,\quad
\gamma(f_{e,M},\widetilde E)=\sup_{k \geq 2} \left| \frac{f_{e,M}^{(k)}(\widetilde{E})}{k!f_{e,M}'(\widetilde{E})}\right|^{\frac{1}{k-1}}
\]
and $\alpha_0=3-2\sqrt{2}\approx 0.1715728$.
\end{defn}

Odell and Gooding~\cite{OG86} compiled a list of starters that have been proposed in the
literature by many authors. The following table provides a formula for those that will be
studied in this paper.
{\renewcommand{\arraystretch}{1.5}
\begin{table}[h]\footnotesize
\begin{tabular}{|c|c|}
  \hline
  \text{Starter} & \text{Formula} \\
  \hline
  \hline
  $S_1$ & $M$ \\
  \hline
  $S_2$ & $M+e\sin(M)$ \\
  \hline
  $S_3$ & $M+e\sin(M)(1+e\cos(M))$ \\
  \hline
  $S_4$ & $M+e$ \\
  \hline
  $S_5$ & $M+\frac{e\sin(M)}{1-\sin(M+e)+\sin(M)}$ \\
  \hline
  $S_6$ & $M+\frac{e(\pi-M)}{1+e}$ \\
  \hline
  $S_7$ & $\min \left\{ \frac{M}{1-e}, S_4, S_6 \right\}$ \\
  \hline
  $S_8$ & $S_3+\frac{e^4 (\pi-S_3)}{20 \pi}$ \\
  \hline
  $S_9$ &  $M+e\sin(M) (1-2e\cos(M)+e^2)^{-\nicefrac12}$\\
  \hline
  $S_{10}$ & $s-\frac{q}{s}$, where $r=\frac{3M}{e}$, $q=\frac{2(1-e)}{e}$
            and $s=[(r^2+q^3)^{\nicefrac{1}{2}} + r]^{\nicefrac13}$ \\
  \hline
\end{tabular}
\bigskip
\caption{Classical starters.}\label{table}
\end{table}
}

In section~\ref{sec-starters} we present an analytical study of the starters
$\tilde{E}=0, \pi, M, \frac{M}{1-e}$ using the notion of approximate zero.  More precisely,
for each of these starters, we obtain in Theorems~\ref{thmE0}, \ref{thmEpi}, \ref{thmEM}
and~\ref{thmEM1-e} regions where they satisfy Smale's $\alpha$-test, thus providing
approximate solutions.  We also show in Theorem~\ref{thmEcubica} that Ng's starter
$S_{10}$~\cite[Eq.~9]{Ng}, which is obtained by solving a cubic equation, gives an
approximate solution on the entire domain.

Similarly, in section \ref{sec-num-starters} we compare the remaining starters $S_2, \ldots,S_{9}$,
and the improved $S_7$ starter obtained by Calvo et al. in~\cite[Prop.~1]{CEMR}. More precisely,
we check numerically where those starters
satisfy Smale's $\alpha$-test on a very fine grid of points in $[0,1) \times [0,\pi]$.

In section \ref{sec-newstarter} we show a simple starter $\tilde{E}=\tilde{E}(e,M)$
which satisfies the $\alpha$-test for all $e\in[0,1)$ and $M\in[0,\pi]$. The starter
is a piecewise-defined function that requires a single cubic root in a small part of
the region close to the corner $e=1$, $M=0$. Apart from that root, the rest of the
expressions involved are constant or rational functions that can be computed with at
most two arithmetic operations. The
highlights of this starter are its computational simplicity and the fact that it is
formally proven to converge at quadratic speed since the first iteration, thus providing
arbitrary precision with a very few Newton's method steps. It should be noted that
reducing the initial error (i.e. the distance from the starter to the exact solution) is
not our design goal.

\begin{thm}\label{thm-starter}
The starter
\[
  \tilde{E}(e,M)=\left\{
  \begin{array}{cl}
    M & \text{if } e \leq \nicefrac12\text{ or } M \geq \nicefrac{2\pi}{3} \\
    \nicefrac{2\pi}{3} & \text{if } e \geq \nicefrac12\text{ and } \nicefrac{\pi}{4}\leq M \leq  \nicefrac{2\pi}{3} \\
    \nicefrac{\pi}{2} & \text{if } e \geq \nicefrac12\text{ and } \nicefrac{\pi}{7}\leq M \leq \nicefrac{\pi}{4} \\
    \frac{M}{1-e} & \text{if } e \geq \nicefrac12,\;M\leq\nicefrac{\pi}{7} \text{ and } M <\frac{ \sqrt[4]{12\alpha}(1-e)^{\nicefrac32}}{\sqrt{e}} \\
    \frac{\sqrt[3]{6Me^2}}{e}-\frac{2(1-e)}{\sqrt[3]{6Me^2}} & \text{otherwise}
  \end{array}
  \right.
\]
is an approximate zero of $f_{e,M}$ for all $e\in[0,1)$ and $M\in[0,\pi]$.
\end{thm}

\begin{figure}[ht!]
\includegraphics[width=0.4\textwidth]{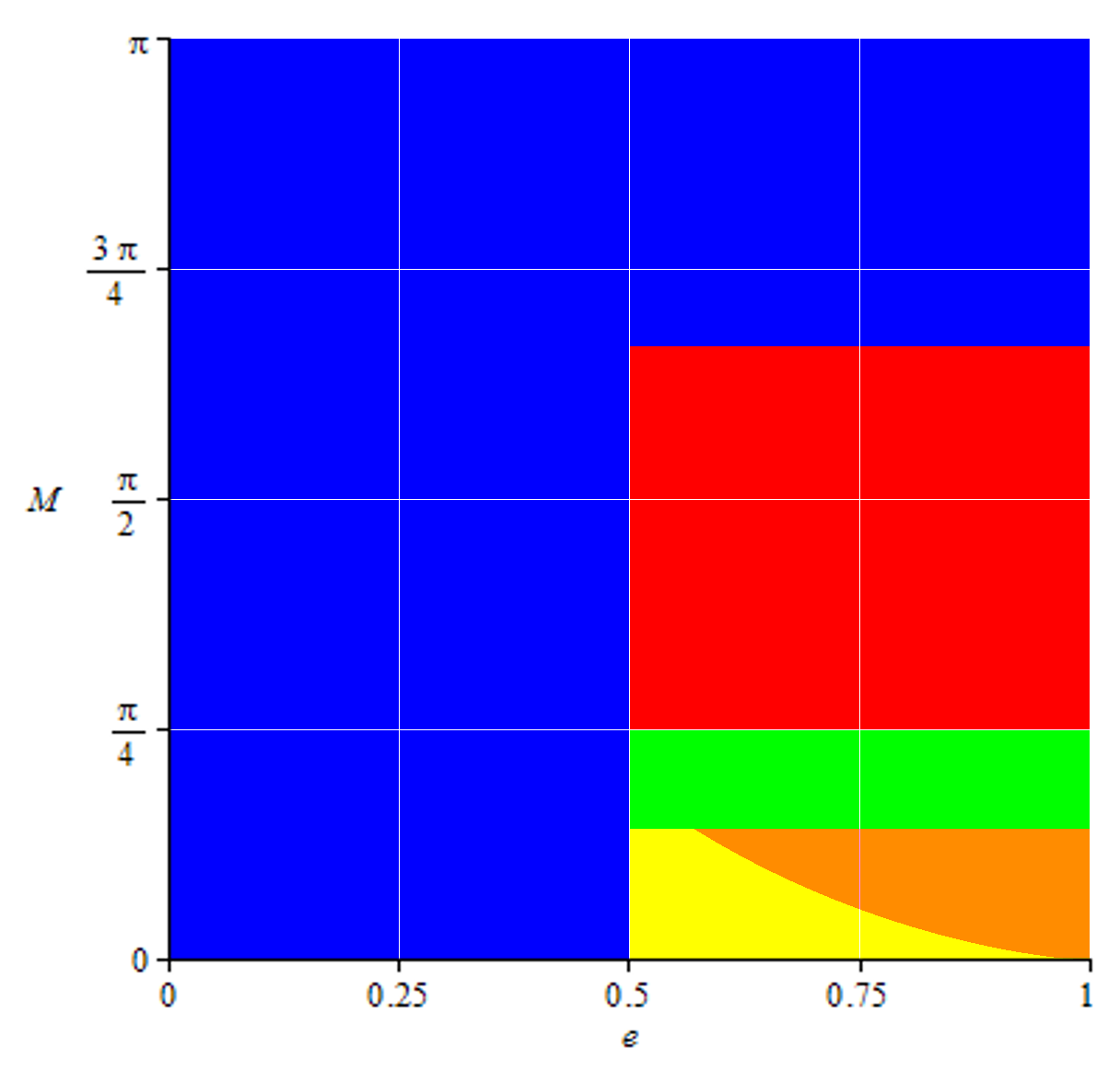}
\caption{The points where $\tilde{E}=M$, $\tilde{E}=\frac{3\pi}2$ and $\tilde{E}=\frac\pi2$ satisfy the $\alpha$-test for
$f_{e,M}(E)$ are shown in blue, red and green respectively. The ones of $\tilde{E}=\frac{M}{1-e}$ and
$\tilde{E}=\frac{\sqrt[3]{6Me^2}}{e}-\frac{2(1-e)}{\sqrt[3]{6Me^2}}$ appear in yellow and orange. }
\label{figEtodo}
\end{figure}

This way of constructing an approximate solution by a piecewise function can be compared
to Ng's approach (see Figure~2 of \cite{Ng}). However, our function is computationally
simpler because Ng's formula outside the corner uses rational functions involving many terms
and near the corner uses $S_{10}$, which requires at least a cubic and a square root
for its computation.

The region near the $(1,0)$ corner where a cubic root is needed can be reduced as much as
desired but cannot be completely avoided, as the following two results show. Other authors
have found similar obstructions in handling values of the eccentricity near $1$ (\cite{Mikk};
\cite{Ng}; \cite{Nije}).

\begin{thm}\label{thm-grid}
For any $\varepsilon>0$, there is a piecewise constant function $\tilde{E}$ defined in
$([0,1) \times [0,\pi]) \setminus ([1-\varepsilon,1] \times [0,\arccos(1-\varepsilon)])$
that satisfies the $\alpha$-test.
\end{thm}

\begin{thm}\label{thm-unavoidable}
Let $\tilde{E}$ be a piecewise rational function in $[0,1) \times [0,\pi]$ with a finite
number of branches defined by polynomial inequalities. Then there exists $(e_0,M_0)$ such
that $\tilde{E}(e_0,M_0)$ is not an approximate zero of $f_{e_0,M_0}$.
\end{thm}

The starter defined in Theorem~\ref{thm-grid} can
be extended if $\varepsilon < 1-\cos(\nicefrac{\pi}{7})$ to the whole region by using
$\frac{M}{1-e}$ and $\frac{\sqrt[3]{6Me^2}}{e}-\frac{2(1-e)}{\sqrt[3]{6Me^2}}$ in the corner,
as in Theorem~\ref{thm-starter}. This result is the basis for a constructing lookup tables
of starters.

Finally, Theorem~\ref{thm-unavoidable} and Remark~\ref{remark-unavoidable} show that the
classical starters $S_1, \ldots, S_8$ and the improved $S_7$ of \cite{CEMR} will necessarily
fail near the corner $(1,0)$, as Figures~\ref{figMM1-e}, \ref{figS234}, \ref{figS567} and~\ref{figS89}
will later illustrate. Our theorem also excludes the possibility of using truncated power
series (with integer exponents) for approximate zeros near the corner.

\section{Analytical study of classical starters via $\alpha$-theory}\label{sec-starters}

In this section we find regions where the starters $\tilde{E}=0,\pi,M,\frac{M}{1-e}$ are approximate zeros of Kepler's equation in  Theorems \ref{thmE0}, \ref{thmEpi}, \ref{thmEM} and \ref{thmEM1-e}. We compare these with the regions computed numerically on a fine grid in Figures \ref{fig0pi} and \ref{figMM1-e}. We also show that Ng's starter $S_{10}$ works in the entire region in Theorem \ref{thmEcubica}.

Throughout the paper, we will need the following technical result.
\begin{lem}\label{lemma_sup}
Let $n\geq 2$ and $x\geq \frac{n!}{(n+1)^{n-1}}$. Then, the sequence
$\{(\frac{x}{k!})^{\frac{1}{k-1}}\}_{k\geq n}$ is decreasing.
\end{lem}
\begin{proof}
It is enough to show that $(\frac{x}{k!})^{\frac{1}{k-1}}\geq(\frac{x}{(k+1)!})^{\frac{1}k}$ for all $k\geq n$, which
is equivalent to the inequality $(\frac{x}{k!})^k\geq( \frac{x}{(k+1)!})^{k-1}$, or more simply
$x\geq \frac{k!}{(k+1)^{k-1}}$.
Note that the sequence $\frac{k!}{(k+1)^{k-1}}$ is decreasing, since
\[
  \frac{(k+1)!(k+1)^{k-1}}{k!(k+2)^k}=\frac{(k+1)^k}{(k+2)^k}<1.
\]
In particular, $x\geq \frac{n!}{(n+1)^{n-1}}\geq \frac{k!}{(k+1)^{k-1}}$ for all $k\geq n$, as we needed.
\end{proof}

\begin{thm}\label{thmE0}
$\tilde{E}=0$ is an approximate zero of $f_{e,M}(E)$ in the
region $R_1\cup R_2$, where
\[
\begin{aligned}
   R_1 &= \left\{0\leq M\leq 4\alpha_0(1-e),\, 0\leq e\leq \frac{3}{11} \right\},\\
   R_2 &= \left\{0\leq M\leq\frac{\sqrt{6}\alpha_0(1-e)^{\nicefrac32}}{\sqrt{e}},\, \frac{3}{11}\leq e<1 \right\}.
\end{aligned}
 \]
\end{thm}
\begin{proof}
It is enough to show that $\alpha(f_{e,M},0)<\alpha_0$, which is equivalent to
\[
  \frac{M}{1-e}\sup_{\doslineas{k\geq 3}{k\;\text{odd}}}\left(\frac{e}{k!(1-e)}\right)^{\frac{1}{k-1}}<\alpha_0,
\]
since $f(0)=-M$, $f'(0)=1-e$, $f^{(\text{even})}(0)=0$ and $f^{(\text{odd})}(0)=\pm e$.
When $e\in[\nicefrac{3}{11},1)$, we have $\frac{e}{1-e}\geq\frac{3}{8}$, and by Lemma~\ref{lemma_sup},
\[
  \sup_{\doslineas{k\geq 3}{k\;\text{odd}}}\left(\frac{e}{k!(1-e)}\right)^{\frac{1}{k-1}}
  = \sqrt{\frac{e}{6(1-e)}}.
\]
In this case, Smale's $\alpha$-test translates into $\frac{M\sqrt{e}}{\sqrt{6}(1-e)^{3/2}}<\alpha_0$,
which corresponds to the region $R_2$. For the remaining case, $e\in[0,\nicefrac{3}{11}]$, we have that
$\frac{e}{1-e}\leq\frac{3}{8}$, so
\[
  \left(\frac{e}{k!(1-e)}\right)^{\frac{1}{k-1}}\leq\left(\frac{1}{16}\right)^{\frac{1}{k-1}}
  \leq\frac{1}{4}\quad\forall\,k\geq 3.
\]
This means that Smale's condition is implied by $\frac{M}{4(1-e)}<\alpha_0$, which corresponds
to the region $R_1$.
\end{proof}

\begin{thm}\label{thmEpi}
$\tilde{E}=\pi$ is an approximate zero of $f_{e,M}(E)$ in the
region $R_3\cup R_4$, where
\[
  \begin{aligned}
   R_3 &= \left\{\pi-4\alpha_0(1+e)<M\leq\pi,\; 0\leq e\leq\frac{3}{5} \right\},\\
   R_4 &= \left\{\pi-\frac{\sqrt{6}\alpha_0(1+e)^{\nicefrac32}}{\sqrt{e}}<M\leq\pi,\;\frac{3}{5}\leq e<1 \right\}.
  \end{aligned}
\]
\end{thm}
\begin{proof}
Since $f(\pi)=\pi-M$, $f'(\pi)=1+e$, $f^{(\text{even})}(\pi)=0$ and $f^{(\text{odd})}(\pi)=\pm e$, Smale's
$\alpha$-test is equivalent to
\[
  \frac{\pi-M}{1+e}\sup_{\doslineas{k\geq 3}{k\;\text{odd}}}
                   \left(\frac{e}{k!(1+e)}\right)^{\frac{1}{k-1}}<\alpha_0.
\]
For any $e\in[0,\nicefrac{3}{5}]$, we have $\frac{e}{1+e}\leq\frac{3}{8}$. This gives the following
estimate for the supremum:
\[
  \left(\frac{e}{k!(1+e)}\right)^{\frac{1}{k-1}}\leq
  \left(\frac{\nicefrac38}{k!}\right)^{\frac{1}{k-1}}\leq
  \left(\frac{1}{16}\right)^{\frac{1}{k-1}}\leq
  \frac{1}{4},\quad\forall\,k\geq 3.
\]
This means that Smale's condition is implied by $\frac{\pi-M}{4(1+e)}<\alpha_0$, which
corresponds exactly to the region $R_3$. For the other case, where $e\in[\nicefrac{3}{5},1)$,
the supremum is $\sqrt{\frac{e}{6(1-e)}}$ by Lemma~\ref{lemma_sup}, so the $\alpha$-condition
is reduced to
\[
  \frac{(\pi-M)\sqrt{e}}{\sqrt{6}(1+e)^{\nicefrac32}} < \alpha_0,
\]
which corresponds to the region $R_4$.
\end{proof}

%
\begin{figure}[ht!]
\includegraphics[width=0.45\textwidth]{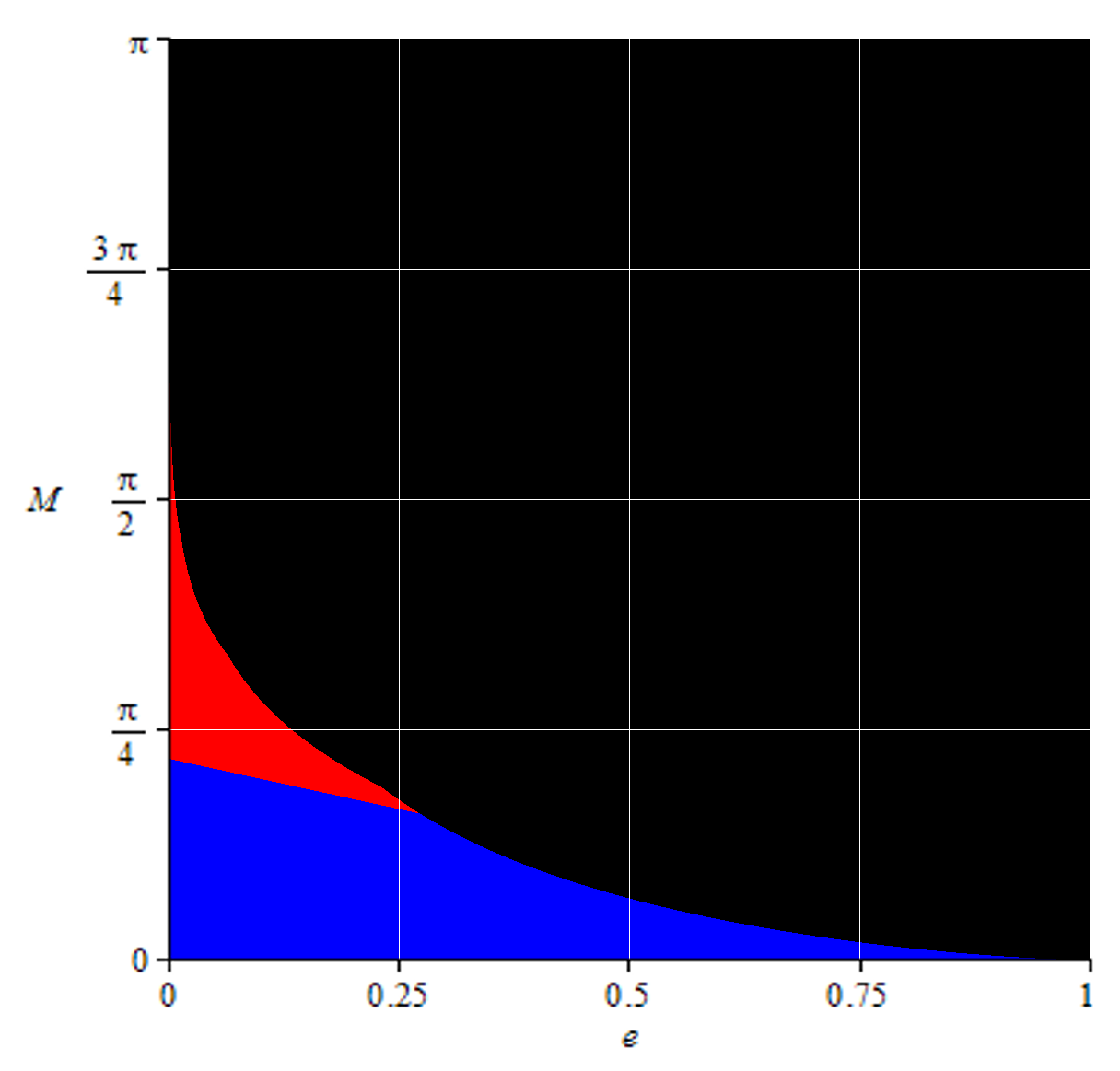}
\includegraphics[width=0.45\textwidth]{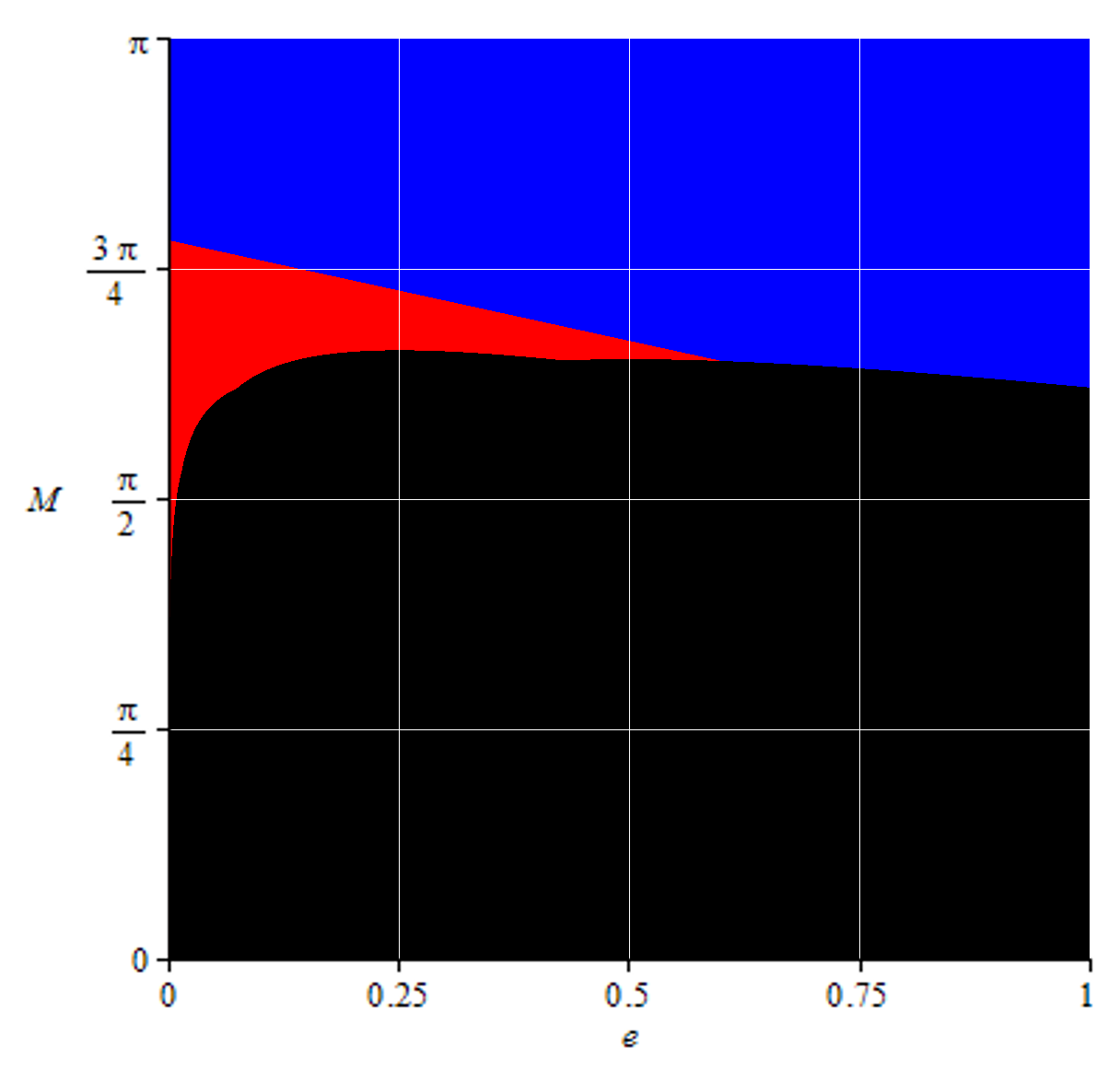}
\caption{The regions of Thm.~\ref{thmE0} and~\ref{thmEpi}
are shown in blue. Red color shows the points where $\tilde{E}=0$ and
$\tilde{E}=\pi$ satisfy the $\alpha$-test for
$f_{e,M}(E)$ that are not in the blue region.}
\label{fig0pi}
\end{figure}

\begin{thm}\label{thmEM}
$\tilde{E}=M$ is an approximate zero of $f_{e,M}(E)$ in the
region
\[
   \left\{0\leq e\leq\frac{1}{2} \right\} \cup \left\{ \frac{2\pi}{3}\leq M\leq\pi \right\}\cup R_2,
\]
where $R_2$ is defined as in Theorem~\ref{thmE0}.
\end{thm}
\begin{proof}
Consider first the strip $M\geq\frac{2\pi}{3}$.
\[
 \beta(f_{e,M},M)=\left|\frac{e\sin(M)}{1-e\cos(M)}\right|\leq \left|\frac{\sin(M)}{1-\cos(M)}\right|=\cot \left(\frac{M}{2} \right)\leq\cot \left(\frac\pi3 \right)=\frac{1}{\sqrt{3}}.
\]
By Lemma~\ref{lemma_sup}, we have that for any even integer $k\geq 2$,
\[
  \left|\frac{e\sin(M)}{k!(1-e\cos(M))}\right|^{\frac{1}{k-1}}\leq\left|\frac{\nicefrac{1}{\sqrt{3}}}{k!}\right|^{\frac{1}{k-1}}\leq\frac{1}{2\sqrt{3}},
\]
and for any odd integer $k\geq 3$,
\[
  \left|\frac{e\cos(M)}{k!(1-e\cos(M))}\right|^{\frac{1}{k-1}}\leq\left|\frac{\nicefrac12}{k!}\right|^{\frac{1}{k-1}}\leq\frac{1}{2\sqrt{3}}.
\]
The last two inequalities together imply $\gamma(f_{e,M},M)\leq\frac{1}{2\sqrt{3}}$ and $\alpha(f_{e,M},M)\leq\nicefrac{1}{6}<\alpha_0$.
This proves that the starter $\tilde{E}=M$ satisfies $\alpha$-test in the strip $M\geq\frac{2\pi}{3}$.

In the region $\left\{ \frac{\pi}{2} \leq M \leq \frac{2\pi}{3}, 0 \leq e \leq \frac12 \right\}$, we have that $\sin(M) \in [\nicefrac{\sqrt3}2,1]$ and $\cos(M) \in [-\nicefrac12,0]$, so
\[
  \beta(f_{e,M},M)=\left| \frac{f(M)}{f'(M)} \right|=\frac{e \sin(M)}{1-e \cos(M)} \leq \frac12.
\]
On the other hand, using Lemma \ref{lemma_sup} gives us
\[
\begin{aligned}
\sup_{\doslineas{k\geq2}{k\;\text{even}}} \left| \frac{f^{(k)}(M)}{k! f'(M)} \right|^{\frac{1}{k-1}} &\leq \sup_{\doslineas{k\geq2}{k\;\text{even}}} \left| \frac{1}{2 k! } \right|^{\frac{1}{k-1}}
=\max \left\{ \frac14, \sup_{\doslineas{k\geq4}{k\;\text{even}}} \left| \frac{1}{2 k! } \right|^{\frac{1}{k-1}} \right\}\\
&=\max \left\{ \frac14, \frac1{\sqrt[3]{48}} \right\}=\frac1{\sqrt[3]{48}} \approx 0.2752,
\end{aligned}
\]
\[
\begin{aligned}
\sup_{\doslineas{k\geq3}{k\;\text{odd}}} \left| \frac{f^{(k)}(M)}{k! f'(M)} \right|^{\frac{1}{k-1}} &\leq \sup_{\doslineas{k\geq3}{k\;\text{odd}}} \left| \frac{1}{4 k! } \right|^{\frac{1}{k-1}}
=\max \left\{ \frac1{\sqrt{24}}, \sup_{\doslineas{k\geq4}{k\;\text{even}}} \left| \frac{1}{4 k! } \right|^{\frac{1}{k-1}} \right\}\\
&=\max \left\{ \frac1{\sqrt{24}}, \frac1{\sqrt[4]{480}} \right\}=\frac1{\sqrt[4]{480}} \approx 0.2136.
\end{aligned}
\]
Therefore, $\gamma(f_{e,M},M) \leq \frac1{\sqrt[3]{48}}$ and the $\alpha$-test holds because $\frac12\frac1{\sqrt[3]{48}} < \alpha_0$.

In the region $\left\{ 0 \leq M \leq \frac{\pi}{2}, 0 \leq e \leq \frac12 \right\}$,
\begin{equation}\label{eq-in1}
\frac{e \sin(M)}{1-e \cos(M)} \leq \frac{\frac12 \sin(M)}{1-\frac12 \cos(M)}\leq\frac{1}{\sqrt{3}}
\end{equation}
and using Lemma \ref{lemma_sup} we obtain that
\[
\begin{aligned}
\sup_{\doslineas{k\geq2}{k\;\text{even}}} \left| \frac{f^{(k)}(M)}{k! f'(M)} \right|^{\frac{1}{k-1}} &\leq
\max \left\{g_2,g_4,\sup_{\doslineas{k\geq6}{k\;\text{even}}} \left| \frac{\frac12 \sin(M)}{k!(1-\frac12 \cos(M))} \right|^{\frac{1}{k-1}} \right\}\\
&\leq \max \left\{g_2,g_4,\sup_{\doslineas{k\geq6}{k\;\text{even}}} \left| \frac{1}{k!} \right|^{\frac{1}{k-1}} \right\} = \max \left\{g_2,g_4,\sqrt[5]{\frac{1}{6!}} \right\},
\end{aligned}
\]
where $g_k=\left(\frac{\frac12 \sin(M)}{k!(1-\frac12 \cos(M))} \right)^{\frac1{k-1}}$ for $k=2,4$.
Similarly,
\[
\begin{aligned}
\sup_{\doslineas{k\geq3}{k\;\text{odd}}} \left| \frac{f^{(k)}(M)}{k! f'(M)} \right|^{\frac{1}{k-1}} &\leq
\max \left\{g_3,g_5,\sup_{\doslineas{k\geq7}{k\;\text{odd}}} \left| \frac{\frac12 \cos(M)}{k!(1-\frac12 \cos(M))} \right|^{\frac{1}{k-1}} \right\}\\
&\leq \max \left\{ g_3,g_5, \sup_{\doslineas{k\geq7}{k\;\text{odd}}} \left| \frac{1}{k!} \right|^{\frac{1}{k-1}} \right\}
=\max \left\{ g_3,g_5, \sqrt[6]{\frac{1}{7!}} \right\},
\end{aligned}
\]
where $g_k=\left(\frac{\frac12 \cos(M)}{k!(1-\frac12 \cos(M))} \right)^{\frac1{k-1}}$ for $k=3,5$.
Therefore,
\[
\gamma(f_{e,M},M) \leq \max \left\{ g_2,g_3,g_4,g_5,\sqrt[5]{\frac{1}{6!}}, \sqrt[6]{\frac{1}{7!}} \right\}=\max \left\{ g_2,g_3,g_4,g_5,\sqrt[5]{\frac{1}{6!}} \right\}.
\]
As an immediate consequence of the second inequality in~\eqref{eq-in1}, we get $g_2 <g_4$,
$\frac{\frac12 \sin(M)}{1-\frac12 \cos(M)} g_4 <\alpha_0$ and $\frac{\frac12 \sin(M)}{1-\frac12 \cos(M)} \sqrt[5]{\frac{1}{6!}} <\alpha_0$.
It remains to see that $\frac{\frac12 \sin(M)}{1-\frac12 \cos(M)} g_k \leq \alpha_0$ for $k=3,5$, which is equivalent to proving
\[
  \frac{\sin^3(M) \cos(M)}{\left( 1-\frac12 \cos(M) \right)^3} <48 \alpha_0^2 \approx 1.41, \text{ and }
  \frac{\sin^4(M) \cos(M)}{\left( 1-\frac12 \cos(M) \right)^5} <3840 \alpha_0^4 \approx 3.33.
\]
In both cases, the left-hand side function has a maximum and the inequalities are true at it.

Finally, note that $f_{e,M}(M)=-e\sin M\leq 0$ and $f_{e,M}$ is increasing, so $0\leq M\leq E$, where
$E$ represents the exact solution of Kepler's equation. In particular, $M$ is always closer to $E$ than
$0$, hence for any point in $R_2$, the starter $\tilde{E}=M$ gives an approximate solution.
\end{proof}

\begin{thm}\label{thmEM1-e}
$\tilde{E}=\frac{M}{1-e}$ is an approximate zero of $f_{e,M}(E)$ in the
region $R_5\cup R_6$, where
\[
   \begin{aligned}
   R_5 &= \left\{0\leq M< \min \left\{ \sqrt[4]{12\alpha_0}\frac{(1-e)^{\nicefrac32}}{e^{\nicefrac12}}, \sqrt[3]{24\alpha_0}\frac{(1-e)^{\nicefrac43}}{e^{\nicefrac13}}  \right\},\, 0\leq e\leq \frac{3}{11} \right\},\\
   R_6 &= \left\{ 0\leq M< \sqrt[4]{12\alpha_0}  \frac{(1-e)^{\nicefrac32}}{e^{\nicefrac12}}  ,\, \frac{3}{11} \leq e<1  \right\}.
   \end{aligned}
 \]
This region contains the region of Theorem~\ref{thmE0}.
\end{thm}
\begin{proof}
In this case we have
\[
  |f(\tilde{E})|=e\left|\frac{M}{1-e}-\sin\left(\frac{M}{1-e}\right)\right|\leq \frac{eM^3}{6(1-e)^3},
\]
$|f'(\tilde{E})|\geq 1-e$ and $|f^{(k)}(\tilde{E})|\leq e$ for all $k\geq 2$. Besides,
\[
\gamma \left( f_{e,M}, \frac{M}{1-e} \right) \leq \max \left\{ \frac{e  \frac{M}{1-e}}{2(1-e)},\, \sup_{k\geq 3}\left|\frac{e}{k!(1-e)}\right|^{\frac{1}{k-1}}  \right\}.
\]
In particular, Smale's $\alpha$-test is satisfied if
\[
\frac{M^4 e^2}{12(1-e)^6} < \alpha_0 \ \ \text{ and } \ \  \frac{eM^3}{6(1-e)^4}\sup_{k\geq 3}\left|\frac{e}{k!(1-e)}\right|^{\frac{1}{k-1}}<\alpha_0.
\]
The first condition is equivalent to $M<\frac{\sqrt[4]{12\alpha_0} (1-e)^{\nicefrac32}}{e^{\nicefrac12}}$, which
is true in both $R_5$ and $R_6$. The second inequality needs to be discussed depending on the value
of $e$.

When $e\in[\nicefrac3{11},1)$, we have by Lemma~\ref{lemma_sup} that
\[
  \sup_{k\geq 3}\left|\frac{e}{k!(1-e)}\right|^{\frac{1}{k-1}}=\sqrt{\frac{e}{6(1-e)}},
\]
so the second inequality becomes $M<\sqrt6 \sqrt[3]{\alpha_0}  \frac{(1-e)^{3/2}}{e^{1/2}}$,
which is automatically true in $R_6$ since $\sqrt6 \sqrt[3]{\alpha_0} > \sqrt[4]{12\alpha_0}$.

In the other case, i.e. when $e\in[0,\nicefrac{3}{11}]$, we have $\frac{e}{1-e}\leq \nicefrac38$. In particular, we can estimate the supremum from above as follows:
\[
  \sup_{k\geq 3}\left|\frac{e}{k!(1-e)}\right|^{\frac{1}{k-1}}\leq
  \sup_{k\geq 3}\left|\frac{3}{8k!}\right|^{\frac{1}{k-1}}=
  \frac14,
\]
where we have used Lemma \ref{lemma_sup}.
Therefore, in the case $e\in[0,\nicefrac{3}{11}]$, the $\alpha$-test is satisfied when
\[
M<\frac{\sqrt[4]{12\alpha_0} (1-e)^{\nicefrac32}}{e^{\nicefrac12}}  \ \ \text{ and } \ \ M < \sqrt[3]{24 \alpha_0} \frac{(1-e)^{\nicefrac43}}{e^{\nicefrac13}},
\]
which is the definition of the region $R_5$.

Finally, the inclusion $R_2 \subseteq R_6$ follows immediately from $\sqrt6 \alpha_0 < \sqrt[4]{12 \alpha_0}$ and $R_1 \subseteq R_5$ from the fact that $4 \alpha_0 (1-e) < \sqrt[4]{12 \alpha_0} \frac{(1-e)^{\nicefrac32}}{e^{\nicefrac12}}$ and $4 \alpha_0 (1-e) < \sqrt[3]{24 \alpha_0} \frac{(1-e)^{\nicefrac43}}{e^{\nicefrac13}}$ for all $e \in [0,\nicefrac3{11}]$.
\end{proof}

\begin{figure}[ht!]
\includegraphics[width=0.45\textwidth]{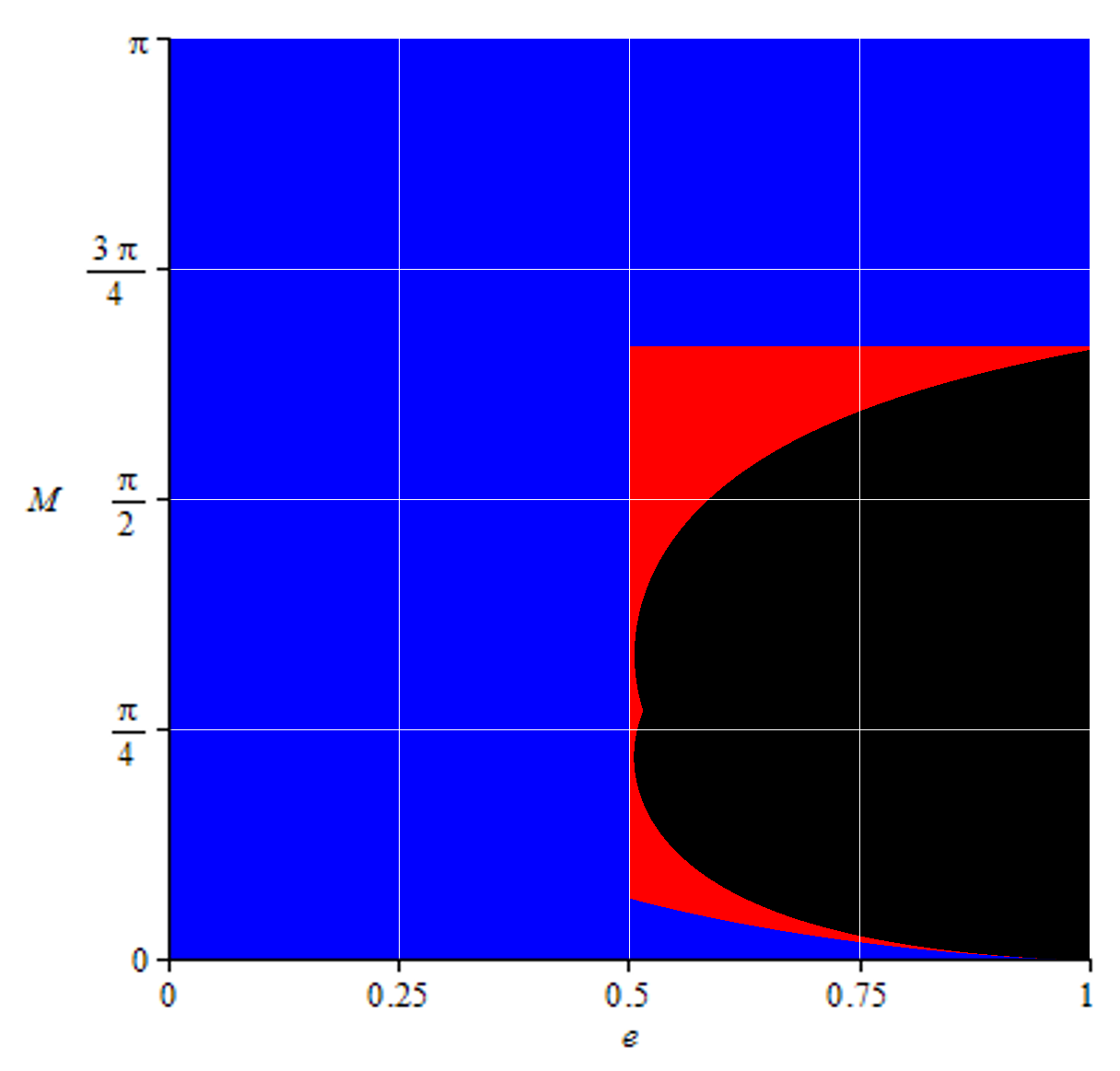}
\includegraphics[width=0.45\textwidth]{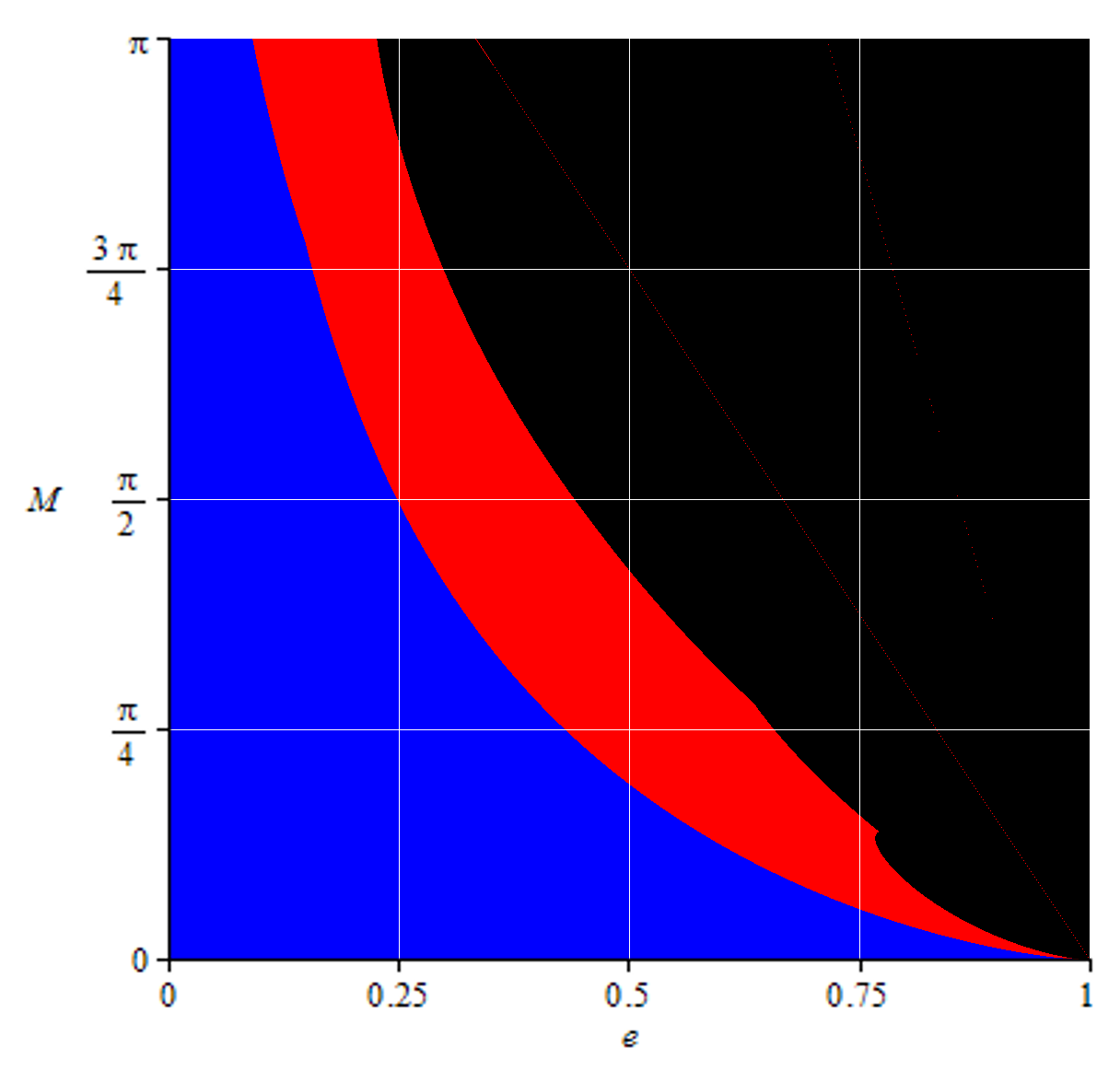}
\caption{The regions of Thm.~\ref{thmEM},  and~\ref{thmEM1-e}
are shown in blue. Red color shows the points where $\tilde{E}=M$ and $\tilde{E}=\frac{M}{1-e}$ satisfy the $\alpha$-test for
$f_{e,M}(E)$ that are not in the blue region.}
\label{figMM1-e}
\end{figure}

\begin{thm}\label{thmEcubica}
The exact solution of the cubic equation $\tilde{E}(1-e)+e\frac{\tilde{E}^3}{6}-M=0$ is an approximate
zero of $f_{e,M}(E)$ in the entire region $[0,1)\times[0,\pi]$.
\end{thm}
\begin{proof}
First, note that the derivative of the left-hand side of the equation is $(1-e)+e\tilde{E}^2/2>0$, so
the expression is increasing. This means that the cubic has only one real root. Moreover, the
values of the cubic at $0$ and $\pi$ are $-M\leq 0$ and $\pi(1-e)+e\frac{\pi^3}{6}-M\geq\pi-M\geq 0$ respectively, so the real root $\tilde{E}$ must be in $[0,\pi]$.
In particular, we have that $\tilde{E}< \sqrt{42}$, so
\[
|f(\tilde{E})| =|\tilde{E}-e \sin(\tilde{E})-M|=\left|\tilde{E}(1-e)+e \left( \frac{\tilde{E}^3}{3!}-\frac{\tilde{E}^5}{5!}+\cdots \right)-M \right| \leq e\frac{\tilde{E}^5}{120}.
\]
Let us now consider two different cases depending on the value of $\tilde{E}$.

If $\tilde{E}\leq \frac{\pi}2$, we have that  $f'(\tilde{E}) \geq 1-\cos(\tilde{E})=2 \sin^2(\frac{\tilde{E}}{2})\geq \frac{4}{\pi^2} \tilde{E}^2$ and
\[
\gamma(f_{e,M},\tilde{E}) \leq \sup_{k\geq2} \left| \frac{1}{k!(1-\cos(\tilde{E})}\right|^{\frac{1}{k-1}} \leq \sup_{k\geq2} \left| \frac{\pi^2}{4 k! \tilde{E}^2  }\right|^{\frac{1}{k-1}} =\frac{\pi^2}{8 \tilde{E}^2}
\]
by Lemma~\ref{lemma_sup}. Therefore, the $\alpha$-test follows if we prove
\[
\frac{e\frac{\tilde{E}^5}{120}}{ \frac{4}{\pi^2} \tilde{E}^2} \frac{\pi^2}{8 \tilde{E}^2} <\frac{\pi^4 \tilde{E}}{3840} < \alpha_0 \Leftrightarrow \tilde{E} <\frac{3840\alpha_0}{\pi^4} \approx 6.76,
\]
which is always true in this region.

If $\tilde{E}> \frac{\pi}2$, then $ \gamma(f_{e,M},\tilde{E}) =\max \{ g_2,g_3,g_4,g_5 \}$, where
\[
\begin{aligned}
g_2 &= \frac{1 }{2 (1-e \cos(\tilde{E}))}, \ g_3=\sqrt{\frac{|\cos(\tilde{E})|}{6(1-e\cos(\tilde{E}))}}, \\
g_4 &=\sup_{\doslineas{k\geq4}{k\;\text{even}}} \left| \frac{1}{k!(1-e \cos(\tilde{E}))}\right|^{\frac{1}{k-1}}= \sqrt[3]{\frac{1}{24(1-e \cos(\tilde{E}))}},\\
g_5 &=\sup_{\doslineas{k\geq5}{k\;\text{odd}}} \left| \frac{1}{k!(1-e \cos(\tilde{E}))}\right|^{\frac{1}{k-1}} =  \sqrt[4]{\frac{1}{120(1-e \cos(\tilde{E}))}}
\leq g_4.
\end{aligned}
\]
Therefore, the $\alpha$-test is satisfied if $\frac{e\tilde{E}^5}{120 (1-e \cos(\tilde{E}))} g_i < \alpha_0$ for $i=2,3,4$.

Since $g_2$, $g_3$ and $g_4$, are increasing in $M$, it is enough to prove the inequalities when
$M=\pi$. Moreover, $\tilde{E}(e,\pi)$ is decreasing, so $\tilde{E}(e,\pi) \in [\sqrt[3]{6 \pi},\pi]$
and $1-e\cos(\tilde{E}(e,\pi)) \geq 1-e\cos(\sqrt[3]{6\pi})$.

We also have that $\pi =e \frac{\tilde{E}^3(e,\pi)}{6} +(1-e)\tilde{E}(e,\pi) \geq e \frac{\tilde{E}^3(e,\pi)}{6} +(1-e)\sqrt[3]{6\pi}$, hence
\begin{equation}\label{eq-Epi}
\tilde{E}(e,\pi) \leq \sqrt[3] {\frac{6 \left(\pi-(1-e)\sqrt[3]{6\pi} \right)}{e}}.
\end{equation}

Let us now study the three different cases.

When $i=2$, it is enough to prove that
\[
\frac{e\tilde{E}^5}{120 (1-e \cos(\tilde{E}))} g_2
<\frac{e\tilde{E}(e,\pi)^5}{240 (1-e \cos(\sqrt[3]{6\pi}))^2}  < \alpha_0,
\]
which is true using that $\tilde{E} \leq \pi$ in $e \in [0,0.17]$, $\tilde{E}(e,\pi) \leq 2.92$ in $e \in [0.17,0.3]$, $\tilde{E}(e,\pi) \leq 2.84$ in $e \in [0.3,0.4]$ and Eq.~\eqref{eq-Epi} in $e\in[0.4,1]$.

When $i=3$, it suffices to show that
\[
\frac{e\tilde{E}^5}{120 (1-e \cos(\tilde{E}))} g_3
<\frac{e\tilde{E}^5(e,\pi) \sqrt{ |cos(\tilde{E}(e,\pi))| }}{120\sqrt{6} (1-e \cos(\sqrt[3]{6\pi}))^{\nicefrac32} }   < \alpha_0,
\]
which is true using that
\begin{itemize}
 \item $\tilde{E} \leq \pi$ and $\sqrt{|cos(\tilde{E}(e,\pi)|} \leq 1$ in $e \in [0,0.2]$,
 \item Eq.~\eqref{eq-Epi} and $\sqrt{|cos(\tilde{E}(e,\pi)|} \leq 1$ in $e \in [0.2,0.7]$,
 \item Eq.~\eqref{eq-Epi} and $\sqrt{|cos(\tilde{E}(e,\pi)|} <0.91 $ in $e\in[0.7,1]$.
\end{itemize}

Lastly, the case $i=4$ follows by using $\tilde{E} \leq \pi$ in $e \in [0,0.2]$ and
Eq.~\eqref{eq-Epi} in $e\in[0.2,1]$.
\end{proof}

\section{Numerical comparison of classical starters via $\alpha$-theory}\label{sec-num-starters}

We tested numerically the $\alpha$-condition on a fine grid (dividing each axis in $1000$ points) for the starters $S_2,\ldots,S_9$, defined in \cite{OG86}, and the improved $S_7$ starter obtained in \cite[Prop.~1]{CEMR}, which we denote $S_{CEMR}$.
Note that none of the starters produce approximate zeros near the corner $(1,0)$.
\begin{figure}[ht!]
\includegraphics[width=0.32\textwidth]{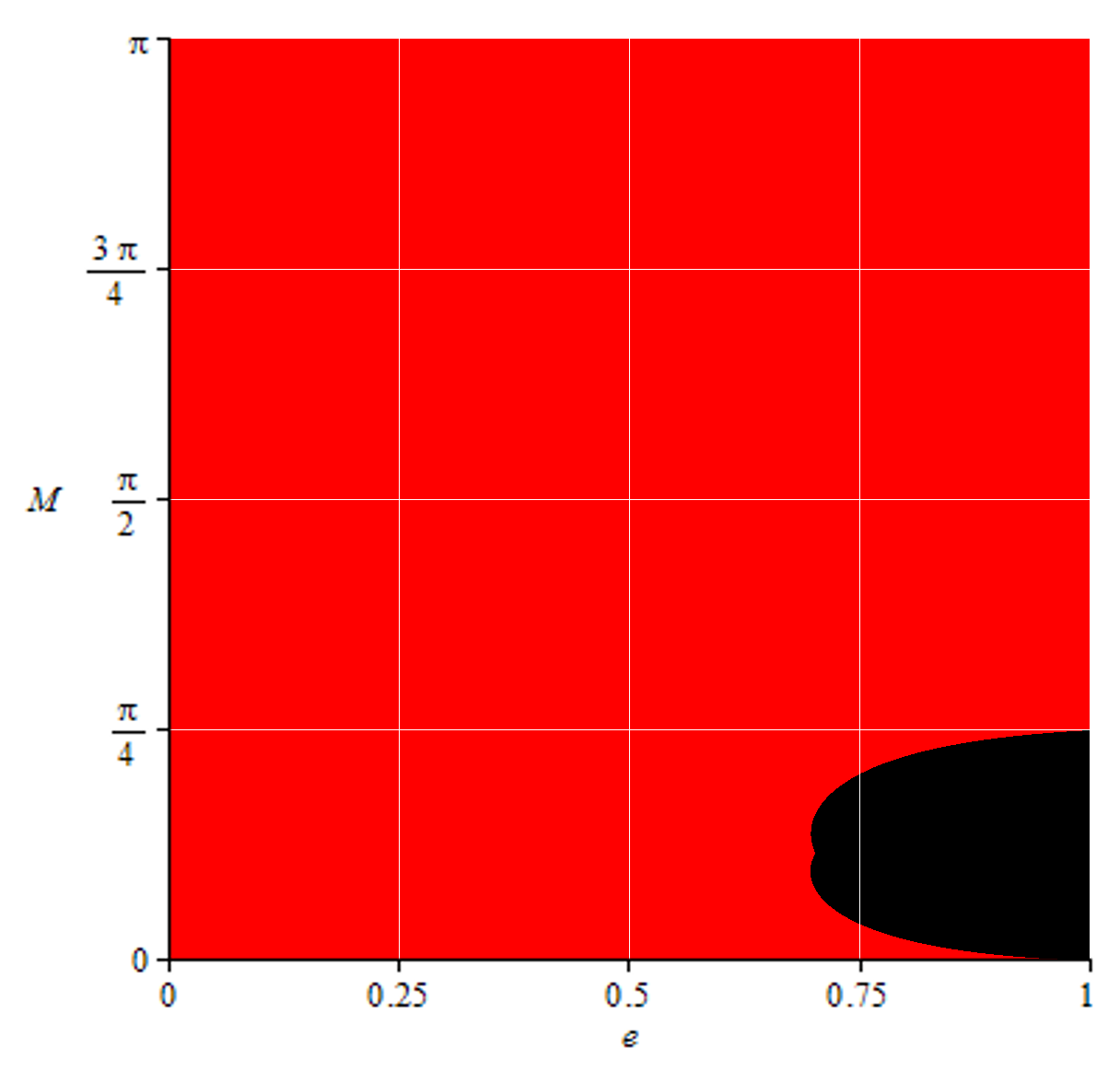}
\includegraphics[width=0.32\textwidth]{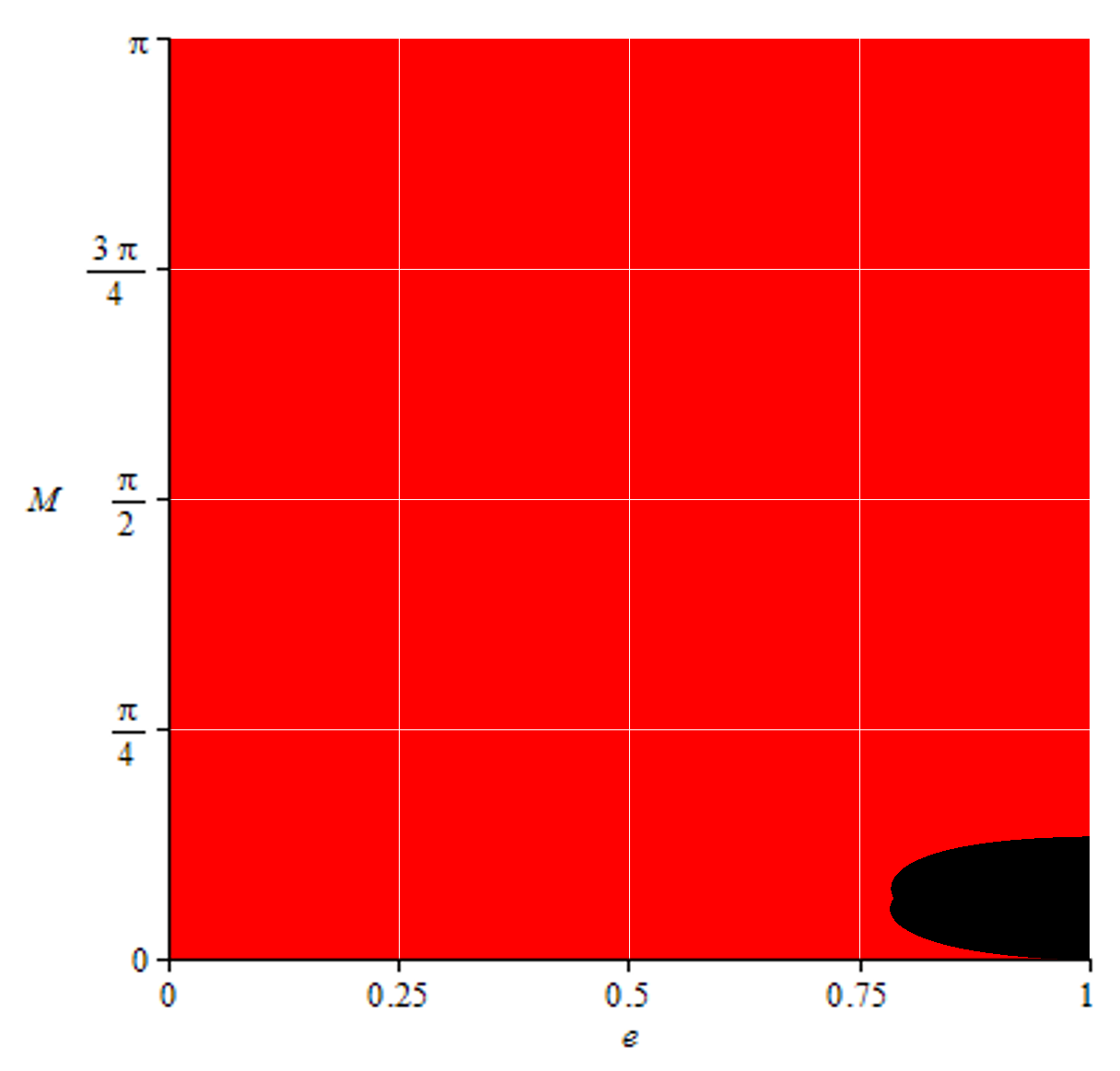}
\includegraphics[width=0.32\textwidth]{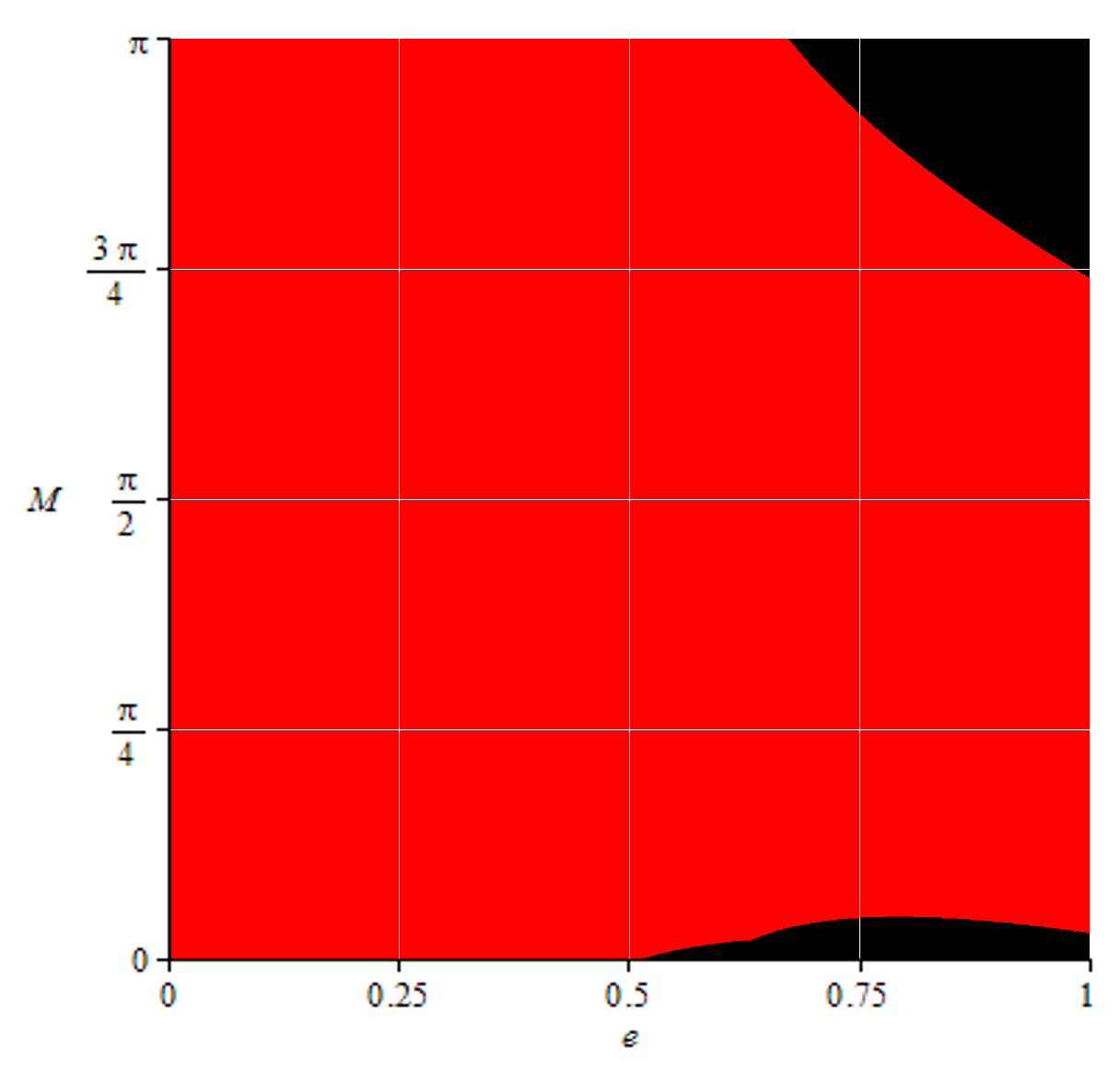}
\caption{The regions of $S_2$, $S_3$ and $S_4$.}
\label{figS234}
\end{figure}

\begin{figure}[ht!]
\includegraphics[width=0.32\textwidth]{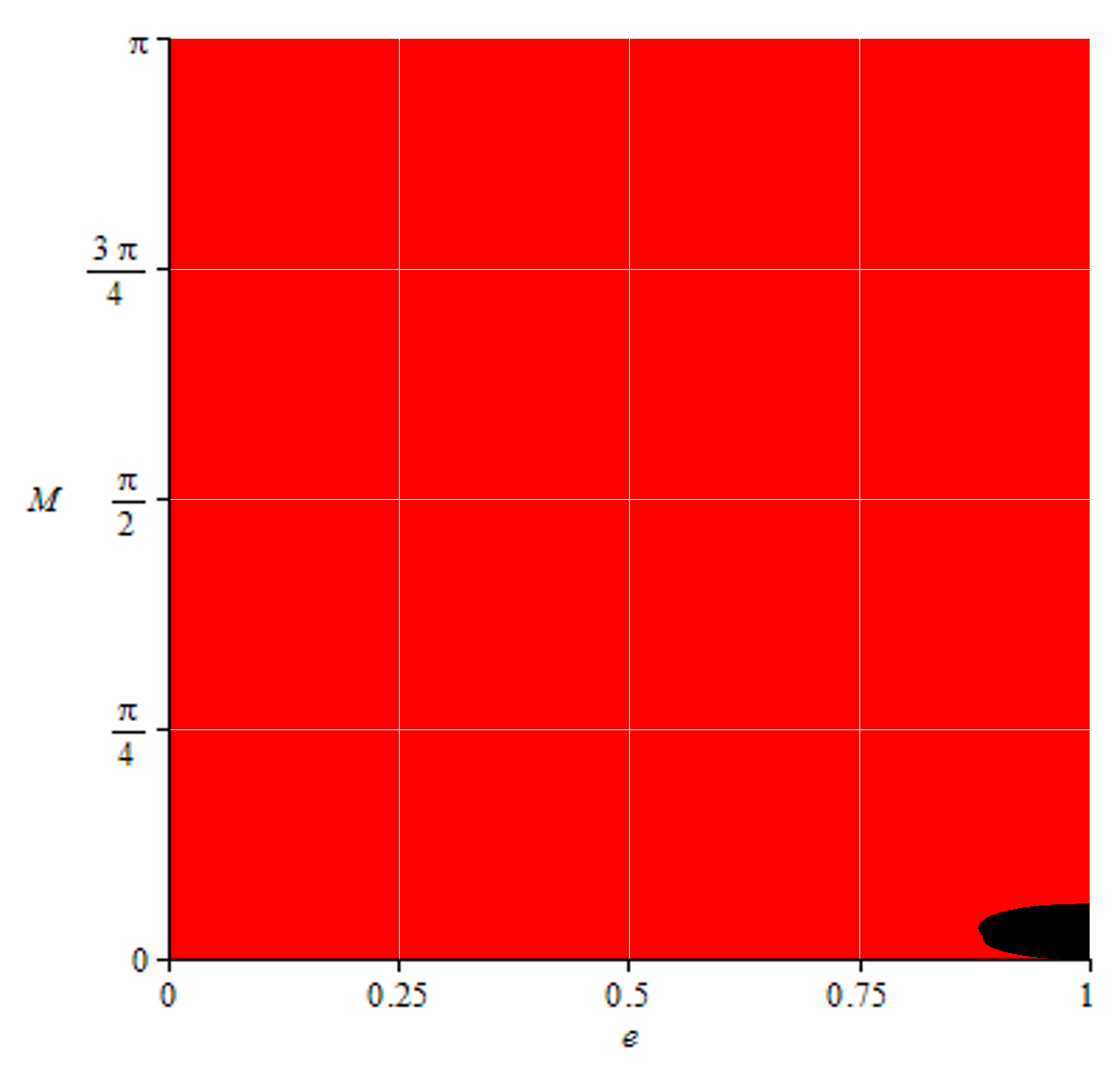}
\includegraphics[width=0.32\textwidth]{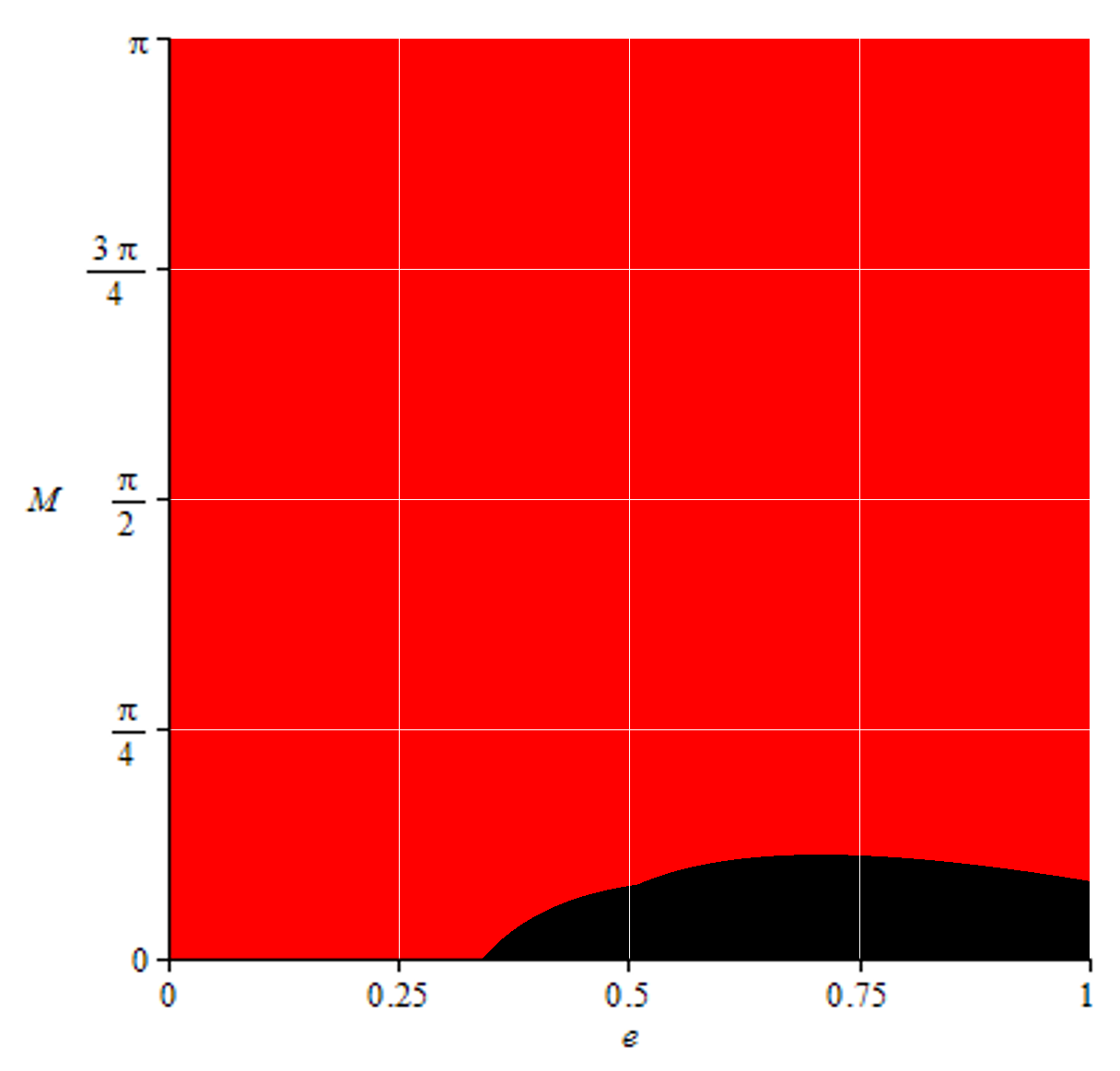}
\includegraphics[width=0.32\textwidth]{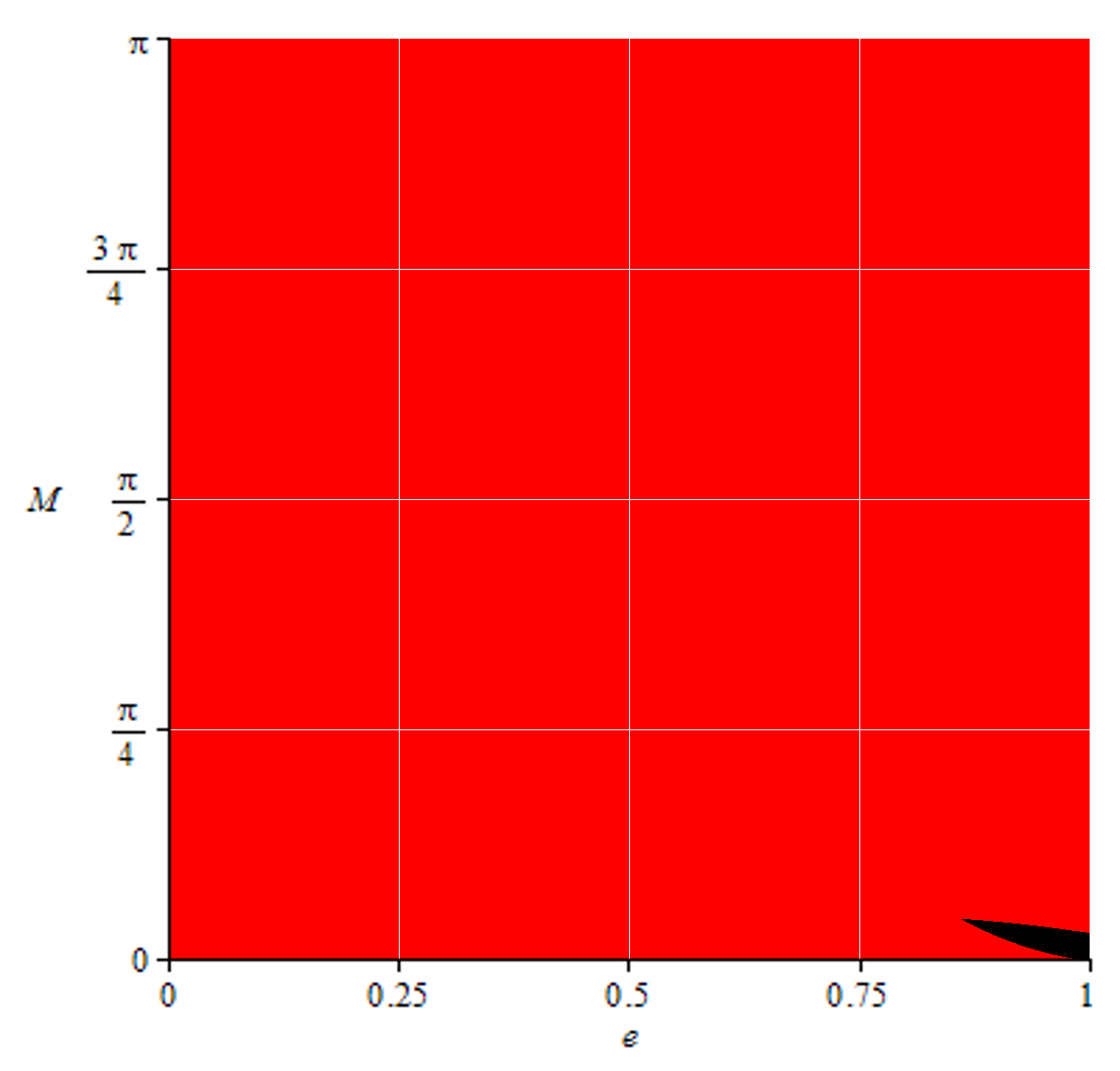}
\caption{The regions of $S_5$, $S_6$ and $S_7$.}
\label{figS567}
\end{figure}

\begin{figure}[ht!]
\centering
\includegraphics[width=0.32\textwidth]{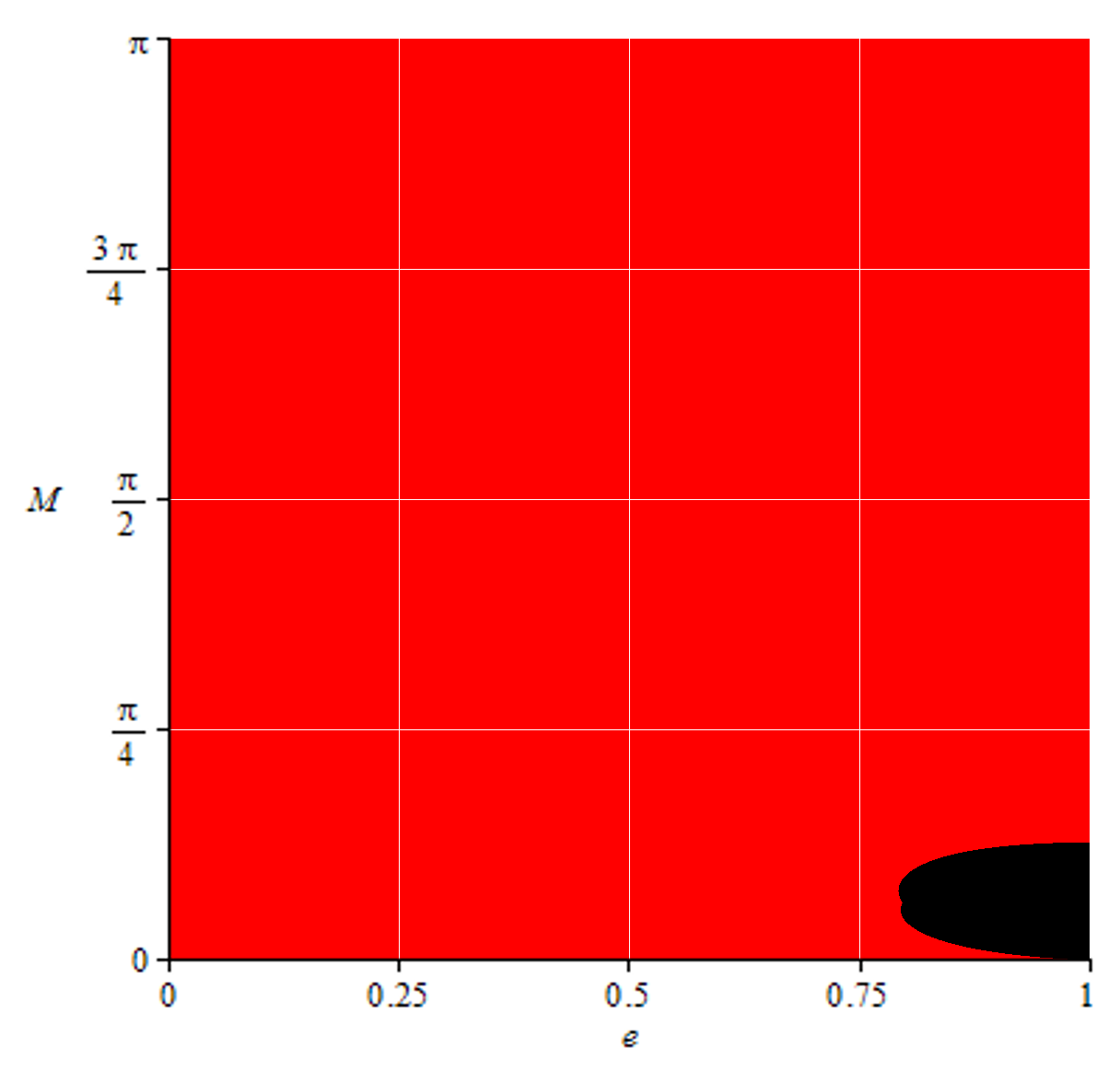}
\includegraphics[width=0.32\textwidth]{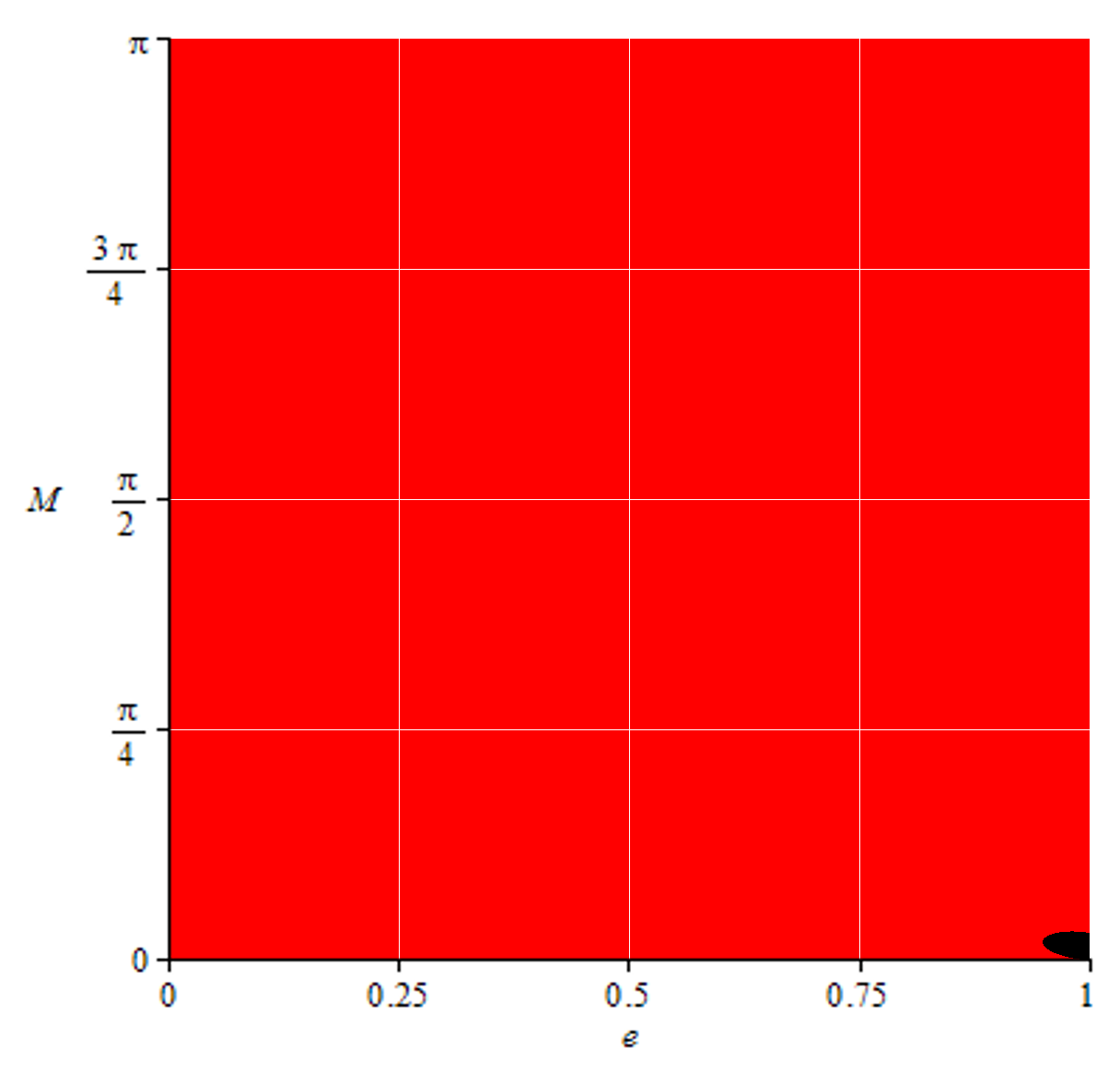}
\includegraphics[width=0.32\textwidth]{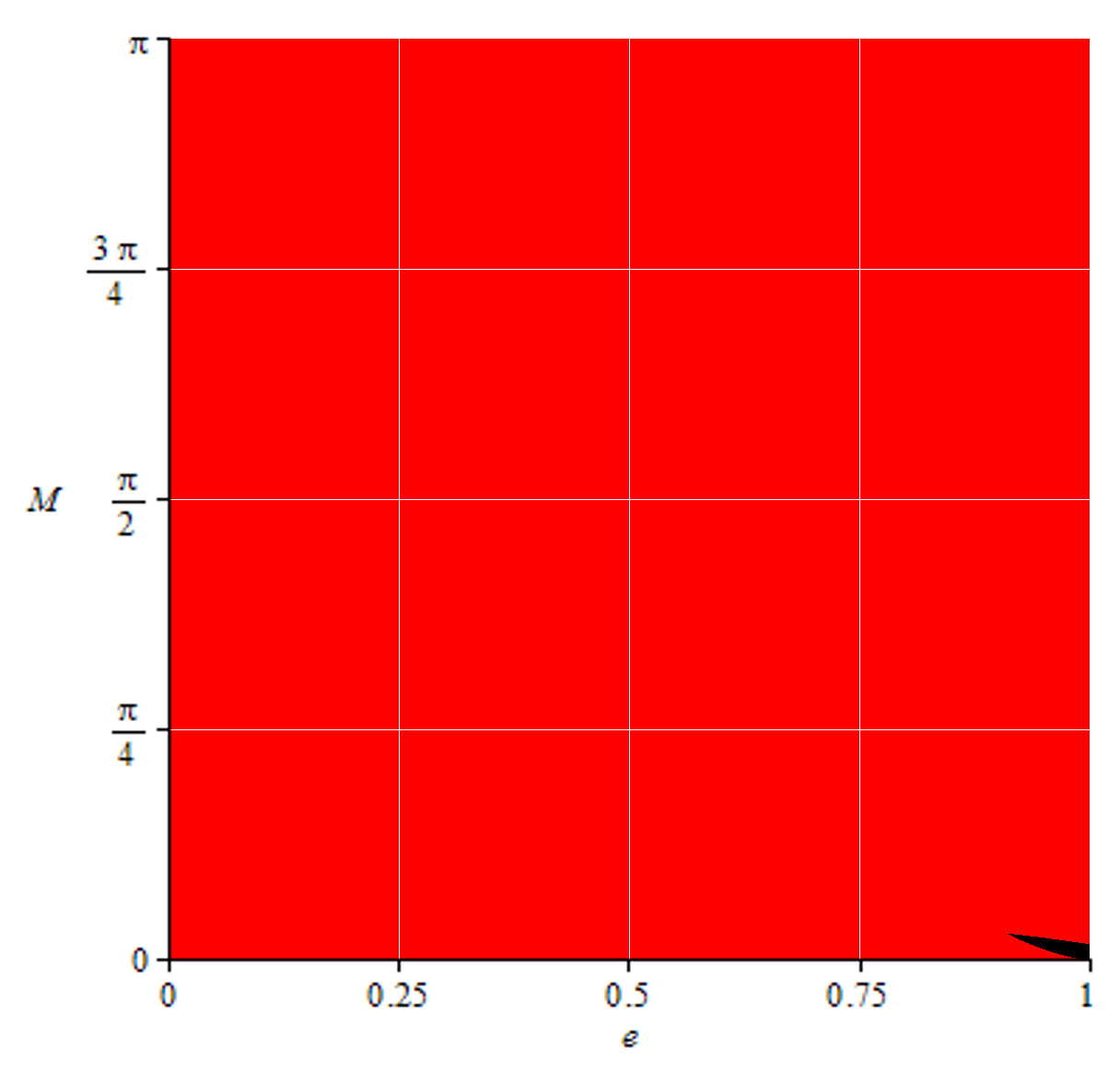}
\caption{The regions of $S_8$, $S_9$ and $S_{CEMR}$.}
\label{figS89}
\end{figure}

\newpage

\section{A simple new starter that covers the entire region}\label{sec-newstarter}

We devote this section to proving Theorem~\ref{thm-starter}. We study each branch separately.

\begin{thm}\label{thmE2pi3}
$\tilde{E}=\frac{2\pi}3$ is an approximate zero of $f_{e,M}(E)$ in the
region
\[
   \left\{\frac{\pi}{4}\leq M\leq \frac{2\pi}{3}   ,\, \frac12 \leq e<1 \right\}.
\]
\end{thm}
\begin{proof}
First of all, we have that
\[
\beta \left(f_{e,M},\nicefrac{2\pi}3 \right) \leq \frac{\frac{2\pi}3-\frac{\sqrt3}2 e-\frac\pi4}{1+\nicefrac{e}{2}}=\frac{\frac{5\pi}{12}-\frac{\sqrt3}2 e}{1+\nicefrac{e}{2}}.
\]
On the other hand,
\[
\begin{aligned}
\gamma \left(f_{e,M},\frac{2\pi}3 \right) &=\max \left\{ \sup_{\doslineas{k \geq 2}{k\;\text{even}}} \left| \frac{e \frac{\sqrt{3}}2}{k! (1+\nicefrac{e}2)} \right|^{\frac1{k-1}},
\sup_{\doslineas{k \geq 3}{k\;\text{odd}}} \left| \frac{ \nicefrac{e}2}{k! (1+\nicefrac{e}2)} \right|^{\frac1{k-1}}  \right\}.
\end{aligned}
\]
Since $\frac{ \nicefrac{e}2}{1+\nicefrac{e}2} \in [\nicefrac13,\nicefrac15]$, we can apply Lemma  \ref{lemma_sup} for $n=4$ and $n=5$:
\[
\begin{aligned}
\sup_{\doslineas{k\geq 2}{k\;\text{even}}} \left| \frac{e \frac{\sqrt{3}}2}{k! (1+\nicefrac{e}2)} \right|^{\frac1{k-1}} &=
\max \left\{   \frac{e \frac{\sqrt{3}}2}{2! (1+\nicefrac{e}2)} , \left( \frac{e \frac{\sqrt{3}}2}{4! (1+\nicefrac{e}2)} \right)^{\frac13}\right\},\\
\sup_{\doslineas{k\geq 3}{k\;\text{odd}}} \left| \frac{\nicefrac{e}{2}}{k! (1+\nicefrac{e}2)} \right|^{\frac1{k-1}} &=
\max \left\{  \left(  \frac{\nicefrac{e}{2}}{3! (1+\nicefrac{e}2)} \right)^\frac12, \left( \frac{\nicefrac{e}{2}}{5! (1+\nicefrac{e}2)} \right)^{\frac14}\right\}.
\end{aligned}
\]
Comparing the four functions, we obtain
\[
\gamma \left(f_{e,M},\frac{2\pi}3 \right) = \left( \frac{e \frac{\sqrt{3}}2}{4! (1+\nicefrac{e}2)} \right)^{\frac13}.
\]
Therefore, the $\alpha$-test is satisfied if
\[
\frac{\frac{5\pi}{12}-\frac{\sqrt3}2 e}{1+\nicefrac{e}{2}} \left( \frac{e \frac{\sqrt{3}}2}{4! (1+\nicefrac{e}2)} \right)^{\frac13} < \alpha_0.
\]
Taking derivatives, it can be shown that the left-hand side of the inequality is a decreasing function of $e$. Also, its value at $e=\nicefrac12$ is approximately $0.1706$, which is less than $\alpha_0$.
\end{proof}

\begin{thm}\label{thmEpi2}
$\tilde{E}=\frac{\pi}2$ is an approximate zero of $f_{e,M}(E)$ in the region
\[
   \left\{\frac{\pi}{7}\leq M\leq \frac{\pi}{4}   ,\, \frac12\leq e<1 \right\}.
\]
\end{thm}
\begin{proof}
We have that $f_{e,M}(\frac{\pi}2 ) =\frac{\pi}2-e-M \leq \frac{\pi}2-e-\frac{\pi}7=\frac{5\pi}{14}-e$
and $f_{e,M}'(\nicefrac{\pi}2 )=1$. Moreover, $f^{(\text{odd})}(\nicefrac{\pi}2)=0$, hence
\[
\gamma( f_{e,M},\nicefrac{\pi}2)=\sup_{\doslineas{k \geq 2}{k\;\text{even}}} \left| \frac{e }{k!} \right|^{\frac1{k-1}}
=\max \left\{ \frac{e}2, \sup_{\doslineas{k \geq 4}{k\;\text{even}}} \left| \frac{e }{k!} \right|^{\frac1{k-1}} \right\}
=\max \left\{ \frac{e}2, \sqrt[3]{\frac{e }{24}} \right\}
\]
by Lemma \ref{lemma_sup}.
The $\alpha$-test is satisfied because
\[
\begin{aligned}
\left(\frac{5\pi}{14}-e\right) \frac{e}2 &\leq \left(\frac{5\pi}{14}-\frac{5\pi}{28}\right) \frac{\nicefrac{5\pi}{28}}{2} \approx 0.1573 <\alpha_0,\\
\left(\frac{5\pi}{14}-e\right) \sqrt[3]{\frac{e}{24}} &\leq \left(\frac{5\pi}{14}-\frac{1}{2}\right) \sqrt[3]{\frac{\nicefrac12}{24}} \approx 0.1711 <\alpha_0,
\end{aligned}
\]
which ends the proof.
\end{proof}

\begin{thm}\label{thmEmartin}
$\tilde{E}=\frac{\sqrt[3]{6Me^2}}{e}-\frac{2(1-e)}{\sqrt[3]{6Me^2}} $ is an approximate zero of $f_{e,M}(E)$ in the region $R_7$, where
\[
   \begin{aligned}
   R_7 &= \left\{  \frac{8(1-e)^{3/2}}{27\sqrt6 \alpha_0 e^{1/2}} < M \leq \frac{\pi}7, \frac{3}{11} \leq e <1 \right\}.
   \end{aligned}
 \]
\end{thm}
\begin{proof}
The first condition we have to impose is that $\tilde{E}\geq0$, which is equivalent to $M \geq \frac{\sqrt2(1-e)^{\nicefrac32}}{3 \alpha_0 e^{\nicefrac12}}$ and true in $R_7$. We also show that $\tilde{E} \leq \nicefrac{\pi}2$ in $[0,\nicefrac{\pi}7] \times [0,1)$, which includes $R_7$.

Indeed, $\tilde{E} \leq \frac{\pi}2$ is equivalent to
\begin{equation}\label{eq-pp2}
h(e,M)= \frac{\sqrt[3]{36} e^{\nicefrac13} M^{\nicefrac23}}{2(1-e)} -\frac{\pi \sqrt[3]{6} e^{\nicefrac23} M^{\nicefrac13}}{4(1-e)} \leq 1.
\end{equation}
For a fixed $e$, the function $h$ has a minimum at $M=\frac{\pi^3}{384}e$ and no other critical points. Therefore, the inequality \eqref{eq-pp2} holds if and only if $h(e,0)  \leq 1$ and $h \left(e,\frac{\pi}7 \right)\leq 1$. The first one is trivial since $h(e,0)=0$ and the second one is equivalent to
\[
2 \sqrt[3]{36} \left(\frac{\pi}{7}\right)^{\frac23} e^{\frac13} -\pi \sqrt[3]{6} \left(\frac{\pi}{7}\right)^{\frac13} e^{\frac23} -4(1-e) <0.
\]
The substitution $e=x^3$ transforms the inequality above into
\[
2 \sqrt[3]{36} \left(\frac{\pi}{7}\right)^{\frac23} x -\pi \sqrt[3]{6} \left(\frac{\pi}{7}\right)^{\frac13} x^2 -4(1-x^3) <0,
\]
which is verified for all $x \in [0,1]$ since the expression in $x$ is increasing and the inequality is true
at $x=1$.

Substituting the expression for $\tilde{E}$ and using the Taylor expansion of $\sin{\tilde{E}}$, we obtain
\[
\begin{aligned}
|f(\tilde{E})| &=|\tilde{E}-e \sin(\tilde{E})-M|=\left|\tilde{E}(1-e)+e \left( \frac{\tilde{E}^3}{3!}-\frac{\tilde{E}^5}{5!}+\cdots \right)-M \right|\\
& \leq \left|\tilde{E}(1-e)+e\frac{\tilde{E}^3}{6}-M \right|+\left|\frac{\tilde{E}^5}{120} \right| =\frac{2(1-e)^3}{9eM}+\left|\frac{\tilde{E}^5}{120} \right|,
\end{aligned}
\]
where we have bounded the alternating series using Leibniz's criterion (possible because $\tilde{E} < \sqrt{42}$).

Since $\tilde{E}\leq \frac{\pi}2$, we have both  $f'(\tilde{E}) \geq 1-e$ and $f'(\tilde{E}) \geq 1-\cos(\tilde{E})=2 \sin^2(\frac{\tilde{E}}{2})\geq \frac{4}{\pi^2} \tilde{E}^2$. Therefore,  the $\alpha$-test follows if we prove the stronger conditions
\begin{equation} \label{eq-qq}
\frac{2(1-e)^2}{9eM} \gamma(f_{e,M},\tilde{E}) < \frac{3\alpha_0}{4} \, \text{ and } \, \left|\frac{\tilde{E}^3\pi^2}{480} \right| \gamma(f_{e,M},\tilde{E}) < \frac{\alpha_0}{4}.
\end{equation}
The second one holds because
\[
\gamma(f_{e,M},\tilde{E}) \leq \sup_{k\geq2} \left| \frac{1}{k!(1-\cos(\tilde{E})}\right|^{\frac{1}{k-1}} \leq \sup_{k\geq2} \left| \frac{\pi^2}{4 k! \tilde{E}^2  }\right|^{\frac{1}{k-1}} =\frac{\pi^2}{8 \tilde{E}^2},
\]
by Lemma \ref{lemma_sup}, and
\[
\left|\frac{\tilde{E}^3\pi^2}{480} \right| \frac{\pi^2}{8 \tilde{E}^2} =\frac{\pi^4 \tilde{E}}{3840} < \frac{\alpha_0}{4} \Leftrightarrow \tilde{E} <\frac{960\alpha_0}{\pi^4} \approx 1.69,
\]
which is true since $\tilde{E} \leq \frac{\pi}2$ in $R_7$.

For the first inequality in \eqref{eq-qq}, we need
\[
\begin{aligned}
\gamma(f_{e,M},\tilde{E}) &\leq   \max \left\{ \frac{e \sin(\tilde{E}) }{2! (1-e)},  \sup_{k\geq3} \left| \frac{e}{k!(1-e)}\right|^{\frac{1}{k-1}}   \right\}\\
& \leq \max \left\{ \frac{e \tilde{E} }{2! (1-e)},  \left| \frac{e}{3!(1-e)}\right|^{\frac{1}{2}}   \right\},
\end{aligned}
\]
true by Lemma \ref{lemma_sup} when $e \geq \nicefrac3{11}$. Therefore,
\[
 \frac{2(1-e)^2}{9eM}  \left| \frac{e}{3!(1-e)}\right|^{\nicefrac{1}{2}}  <\frac{3\alpha_0}4 \Leftrightarrow  M > \frac{8(1-e)^{\nicefrac32}}{27\sqrt6 \alpha e^{\nicefrac12}},
\]
which is one of the conditions of the region $R_7$.

It only remains to show that
\[
\frac{2(1-e)^2}{9eM}  \frac{e \tilde{E} }{2! (1-e)}=\frac{\tilde{E}(1-e)}{9M}  <\frac{3\alpha_0}4,
\]
which is equivalent to
\begin{equation*}
M-\frac{4}{27\alpha_0}(1-e)\tilde{E} >0 \text{ or } \underbracket[0.5pt]{\frac{\sqrt[3]{6} e^{\nicefrac23}}{(1-e)^2} M^{\nicefrac43} -\frac{4\sqrt[3]{36} e^{\nicefrac13}}{27\alpha_0(1-e)} M^{\nicefrac23}}_{g(e,M)} >-\frac{8}{27 \alpha_0}.
\end{equation*}
This is true for every $e \in [0,1)$ and $M \in [0,\pi]$ because, if we fix $e$, the function $g$ has a minimum at $M=\sqrt{\frac{48}{27^3 \alpha_0^3}}$ and
\[
g \left(e,\sqrt{\frac{48}{27^3 \alpha_0^3}} \right)=-\frac{24}{27^2 \alpha_0^2} > -\frac{8}{27 \alpha_0}.
\]
\end{proof}

\begin{proof}[Proof of Theorem~\ref{thm-starter}]
It follows immediately from Theorems \ref{thmEM}, \ref{thmE2pi3}, \ref{thmEpi2}, \ref{thmEM1-e} and \ref{thmEmartin}, and the inequality $\sqrt[4]{12 \alpha_0} > \frac{8}{27 \sqrt6 \alpha_0}$ that implies that the ``otherwise'' region is included in the one from Theorem \ref{thmEmartin}.
\end{proof}

\section{Approximate solutions near $e=1$ and $M=0$}\label{sec-corner}

In this section we will prove Theorems~\ref{thm-grid} and \ref{thm-unavoidable}.

\begin{proof}[Proof of Theorem~\ref{thm-grid}]
Given $\varepsilon >0$, let us take a natural number $N$ such that $N> \frac{\pi+2}{2 \alpha_0 \varepsilon^2}$. Given two integers $i \in \{ 0,\ldots,N-1 \}$ and $j \in \{ 0,\ldots, N \}$, we define the constants $E_{ij}^{\text{low}}=\frac{\pi j}N$ and $E_{ij}^{\text{up}}=\pi$, which satisfy
\begin{align*}
E_{ij}^{\text{low}}-\frac{i}N \sin (E_{ij}^{\text{low}})-\frac{\pi j}N  &=-\frac{i}N \sin \left(\frac{\pi j}N \right) \leq 0,\\
E_{ij}^{\text{up}}-\frac{i}N \sin (E_{ij}^{\text{up}})-\frac{\pi j}N    &=\pi-\frac{\pi j}N \geq0,
\end{align*}
respectively. By the bisection method, we can thus find $E_{ij}$ such that
\begin{equation*}
\frac{\pi j}N=E_{ij}^{\text{low}} \leq E_{ij} \leq E_{ij}^{\text{up}}=\pi \; \text{ and } \; \left| E_{ij}-\frac{i}N \sin (E_{ij})-\frac{\pi j}N \right| <\frac{1}{N}.
\end{equation*}

Given $(e,M) \in ([0,1) \times [0,\pi]) \setminus ([1-\varepsilon,1] \times [0,\arccos(1-\varepsilon)])$, we now define $\tilde{E}(e,M)=E_{ij}$, where $i=\lfloor Ne\rfloor \in \{ 0,\ldots,N-1 \}$ and $j=\left\lceil \frac{M N}{\pi} \right\rceil \in \{ 0,\ldots, N \}$. Therefore, $\tilde{E}$ is a piecewise constant function and it only remains to show that it satisfies the $\alpha$-test.

Indeed, we have that
\begin{align*}
|f(\tilde{E})|&=|E_{ij} -e \sin (E_{ij})-M|\\
&=\left|\left(E_{ij}-\frac{i}N \sin (E_{ij})-\frac{\pi j}N \right) -\left(e-\frac{i}N \right) \sin (E_{ij})-\left(M-\frac{\pi j}N \right) \right|\\
&< \frac1N+\left| e-\frac{i}N \right| +\left| M-\frac{\pi j }N \right| \leq \frac{\pi+2}{N}.
\end{align*}
On the other hand, $|f'(\tilde{E})|=1-e\cos(\tilde{E}) \geq \varepsilon$ because
\[
|f'(\tilde{E})| \geq
\left\{
\begin{array}{ll}
1-e \geq \varepsilon & \text{ if } e \in [0,1-\varepsilon],\\
1-\cos(E_{ij}) \geq 1-\cos(M) \geq \varepsilon & \text{ if } \tilde{E} \in [0,\nicefrac{\pi}2], M \geq \arccos(1-\varepsilon),\\
1 \geq \varepsilon   & \text{ if } \tilde{E} \in [\nicefrac{\pi}2,\pi],
\end{array}
\right.
\]
where we have used that $E_{ij} \geq E_{ij}^{\text{low}}=\frac{\pi j}{N} = \frac{\pi \lceil \frac{MN}{\pi}\rceil}{N}\geq M$.

Since $|f^{(k)}(\tilde{E})| \leq 1$, we obtain using Lemma \ref{lemma_sup} and the hypothesis over $N$ that
\[
\alpha(f_{e,M},\tilde{E}) \leq \frac{\pi+2}{N \varepsilon} \sup_{k\geq2} \left| \frac{1}{k! \varepsilon} \right|^{\frac{1}{k-1}}
\leq \frac{\pi+2}{2 N \varepsilon^2} < \alpha_0,
\]
which ends the proof.
\end{proof}

\begin{proof}[Proof of Theorem~\ref{thm-unavoidable}]
We proceed by contradiction, i.e.~we assume that $\tilde{E}(e,M)$ is an approximate zero of
$f_{e,M}$ for all $e\in [0,1)$ and $M\in[0,\pi]$. Since the branches of $\tilde{E}$ are given
by polynomial inequalities, there is an open set $U\subseteq\R^2$ and $\varepsilon>0$ such that
$\overline{U}\supset \{1\}\times[0,\varepsilon]$ and $\tilde{E}$ is a rational function on $U\cap([0,1)\times[0,\pi])$. We also assume that $U\subseteq [\nicefrac{1}{2},1)\times[0,0.0001]$.

By definition of approximate zero, we have that
\[
  |f(\tilde{E})|
  \underbracket[0.5pt]{
  \max\left\{
     \frac{e |\sin(\tilde{E})|}{2(1-e\cos\tilde{E})^2},
     \sqrt{\frac{e|\cos(\tilde{E})|}{6(1-e\cos\tilde{E})^3}},
     \sqrt[3]{\frac{e|\sin(\tilde{E})|}{24(1-e\cos\tilde{E})^4}}
   \right\}}_{B}<\alpha_0.
\]
It can be readily verified that $B\geq 0.14433$ for all $e\in [\nicefrac{1}{2},1)$ and any
$\tilde{E}\in\R$, so $|f(\tilde{E})|<\frac{\alpha_0}{0.14433}\leq 1.1888$ in $U$. By the
triangle inequality, this implies that $|\tilde{E}|< 1.1888+e+M<2.1889$ in $U$. Repeating
the argument, but using that $|\tilde{E}|<2.1889$, it can be shown that $B\geq 0.176$,
so $|\tilde{E}|<\frac{\alpha_0}{0.176}+e+M\leq 1.975$ in $U$. Doing this one more time, gives $B\geq 0.2368$
and the estimate $|\tilde{E}|<1.725$ in $U$.

Since $\tilde{E}$ is bounded in $U$, it can be extended analytically to $\{1\}\times(0,\delta)$
for some $0<\delta<\varepsilon\leq 0.0001$. To show this, recall that $\tilde{E}(e,M)=\frac{p(e,M)}{q(e,M)}$ for some polynomials $p$ and $q$ with no common factors. Now, if $q(1,M)$ were zero (as a polynomial),
then $q$ would be divisible by $e-1$ and $p$ would not, so $\tilde{E}$ would not be bounded, in
contradiction with our previous result. This proves that $q(1,M)\not\equiv 0$, so we can
take $\delta>0$ small enough to ensure that $q(1,M)$ has no roots in $(0,\delta)$, hence $\tilde{E}(1,M)$
is well defined.

Denote $\tilde{E}_1(M)=\tilde{E}(1,M)$ for $M\in(0,\delta)$. Using that $B\geq\frac{e |\sin(\tilde{E})|}{2(1-e\cos\tilde{E})^2}$, we get
\[
  |\tilde{E}-e\sin \tilde{E}-M|\leq
  \frac{2\alpha_0(1-e\cos\tilde{E})^2}{e |\sin(\tilde{E})|}.
\]
Taking limit as $e\rightarrow 1^-$, we obtain
\[
  |\tilde{E}_1-\sin \tilde{E}_1-M|\leq
  \frac{2\alpha_0(1-\cos\tilde{E}_1)^2}{|\sin(\tilde{E}_1)|}=
  \frac{4\alpha_0|\sin(\frac{\tilde{E}_1}{2})|^3}{|\cos(\frac{\tilde{E}_1}{2})|}\leq
  \frac{\alpha_0|\tilde{E}_1|^3}{2|\cos(\frac{\tilde{E}_1}{2})|}<
  0.133 |\tilde{E}_1|^3
\]
for all $M\in(0,\delta)$. By the power series expansion of $\sin(\tilde{E}_1)$,
\[
  \left|\frac{\tilde{E}_1^3}{3!}-\frac{\tilde{E}_1^5}{5!}+\ldots-M\right|< 0.133 |\tilde{E}_1^3|.
\]
By the triangle inequality,
\[
\begin{aligned}
  \left|\frac{\tilde{E}_1^3}6-M\right| &\leq
  0.133 |\tilde{E}_1^3| + \left|\frac{\tilde{E}_1^5}{5!}-\frac{\tilde{E}_1^7}{7!}+\ldots\right|\\
  &\leq
  |\tilde{E}_1^3|\left(0.133+\frac{\tilde{E}_1^2}{120} \left(1+\frac{\tilde{E}_1^2}{6 \cdot 7}+\frac{\tilde{E}_1^4}{6 \cdot 7 \cdot 8 \cdot 9}+\ldots \right)\right)\\
  &\leq
  |\tilde{E}_1^3|\left(0.133+\frac{1.725^2}{120} \left(1+\frac{1.725^2}{6^2}+\frac{1.725^4}{6^4}+\ldots \right)\right)\\
  & \leq 0.161|\tilde{E}_1^3|,
\end{aligned}
\]
for all $M\in(0,\delta)$. This implies that $(\nicefrac16-0.161)|\tilde{E}_1^3|\leq M$, or
equivalently,
\[
  |\tilde{E}_1|\leq \sqrt[3]{\frac{|M|}{\nicefrac16-0.161}}\xrightarrow[M\to 0^+]{}0.
\]
This shows that $\tilde{E}_1$ has a removable singularity at $M=0$, so it can be
extended analytically to $[0,\delta)$ with $\tilde{E}_1(0)=0$. Moreover,
$\tilde{E}_1(M)=Mr(M)$ for some analytic function $r(M)$ in $[0,\delta)$, since the
power series of $\tilde{E}_1$ cannot have a non-zero constant term.

Finally, by definition of approximate zero,
\[
  \begin{aligned}
  \alpha_0 &>
  \frac{|f(\tilde{E})|}{1-e\cos(\tilde{E})}\max\left\{
      \frac{e |\sin(\tilde{E})|}{2(1-e\cos\tilde{E})},
     \sqrt{\frac{e|\cos(\tilde{E})|}{6(1-e\cos\tilde{E})}}
  \right\}\\
  & \geq
  \frac{|f(\tilde{E})|}{1-e\cos(\tilde{E})}\max\left\{
      \frac{e|\sin(\tilde E)|}{\sqrt{6(1-e\cos(\tilde E))}},
      \frac{e|\cos(\tilde E)|}{\sqrt{6(1-e\cos(\tilde E))}}
  \right\}\\
   &= \frac{e|f(\tilde E)|}{\sqrt{6}(1-e\cos(\tilde{E}))^{\nicefrac32}}\max\{
                                |\sin \tilde E|, |\cos \tilde E|\}\geq
   \frac{e|f(\tilde E)|}{\sqrt{12}(1-e\cos(\tilde{E}))^{\nicefrac32}},
  \end{aligned}
\]
and taking limit as $e\rightarrow 1^-$,
\[
  \begin{aligned}
  |\tilde{E}_1-\sin \tilde{E}_1-M| &\leq
  \sqrt{12}\alpha_0(1-\cos(\tilde{E}_1))^{\nicefrac32}=\\
  &=\sqrt{96}\alpha_0|\sin^3 \left(\frac{\tilde{E}_1}{2} \right)|\leq
  \frac{\sqrt{96}\alpha_0|\tilde{E}_1^3|}{8}=\sqrt{\frac32}\alpha_0|\tilde{E}_1|^3.
  \end{aligned}
\]
Dividing by $M$, using that $\tilde{E}_1(M)=Mr(M)$ and taking limits as $M\rightarrow 0^+$,
\[
  \left|r(M)-\frac{\sin(M r(M))}{M}-1\right|
  \leq   \sqrt{\dfrac32}\alpha_0 M^2 |r(M)|^3,
\]
which gives us the contradiction $1 \leq 0$.
\end{proof}

\begin{remark}\label{remark-unavoidable}
Note that in the proof of Theorem~\ref{thm-unavoidable} we use the rationality of the function only to show that it can be analytically extended to a small segment $\{1\}\times[0,\varepsilon]$ for some $\varepsilon >0$. If we start with an analytic function defined on $[0,1] \times [0,\pi]$, this step is not necessary and the same contradiction is obtained.

This shows that the classical starters $S_1,\ldots,S_8$, as well as $S_{CEMR}$, are not approximate zeros in the entire domain, as Figures~\ref{figMM1-e}, \ref{figS234}, \ref{figS567} and \ref{figS89} illustrate.
\end{remark}

\section*{Acknowledgements}
The authors would like to thank Prof. Antonio Elipe for his valuable help.

\end{document}